\documentclass[format=acmsmall, review=false,screen=true]{acmart}

\startPage{1}
\fancypagestyle{firstpagestyle}{
	\fancyhf{}
}
\fancypagestyle{acmec}{
	\fancyhf{}
}
\renewcommand{\footnotetextcopyrightpermission}[1]{\thankses}

\usepackage{booktabs} 

\usepackage[ruled]{algorithm2e} 

\SetAlFnt{\small}
\SetAlCapFnt{\small}
\SetAlCapNameFnt{\small}
\SetAlCapHSkip{0pt}
\IncMargin{-\parindent}

\acmISBN{978-1-4503-5829-3/18/06}

\setcopyright{acmlicensed}

\acmDOI{10.1145/3219166.3219177}

\copyrightyear{2018}
\acmYear{2018} 

\acmConference[ACM EC '18]{2018 ACM Conference on Economics and Computation}{June 18--22, 2018}{Ithaca, NY, USA}
\acmPrice{15.00}

\usepackage{graphicx} 
\usepackage{color}
\usepackage{amsthm}
\usepackage{amssymb}
\usepackage{tikz}
\usepackage{enumitem}
\usepackage[]{url} 

\newtheorem{theorem}{Theorem}
\newtheorem{prop}[theorem]{Proposition}
\newtheorem{lemma}[theorem]{Lemma}
\newtheorem{cor}[theorem]{Corollary}

\theoremstyle{definition}
\newtheorem{example}[theorem]{Example}
\theoremstyle{remark}

\numberwithin{theorem}{section}

\newenvironment{proofsketch}{%
  \renewcommand{\proofname}{Proof Sketch}\proof}{\endproof}
  
\newcommand{\eps}{\varepsilon}
\newcommand{\paren}[1]{\left(#1\right)}
\newcommand{\bracket}[1]{\left[#1\right]}
\newcommand{\set}[1]{\left\{#1\right\}}

\newcommand{\abs}[1]{\left\lvert#1\right\rvert}
\newcommand{\norm}[1]{\left\lVert#1\right\rVert}
\newcommand{\mb}{\mathbb}
\newcommand{\mc}{\mathcal}

\newcommand{\R}{\mathbb{R}}
\newcommand{\Z}{\mathbb{Z}}
\newcommand{\N}{\mathcal{N}}

\DeclareMathOperator*{\esssup}{ess\,sup}

\newcommand{\new}[1]{#1}

\usepackage{multiaudience}

\SetNewAudience{full}
\SetNewAudience{abb}
\DefCurrentAudience{full}

\begin{document}

\title[Naive Bayesian Learning in Social Networks]{Naive Bayesian Learning in Social Networks}

\thanks{\emph{Full version of the paper published in Proceedings of the 19th ACM Conference on Economics and Computation (EC'18) \citep{proceedingsversion}.}}

\author{Jerry Anunrojwong}
\affiliation{%
  \institution{Harvard University}
  \department{Department of Statistics}
  \city{Cambridge}
  \state{MA}
  \postcode{01238}
  \country{USA}
}
\email{jerryanunroj@gmail.com}
\author{Nat Sothanaphan}
\affiliation{%
  \institution{Massachusetts Institute of Technology}
  \department{Department of Mathematics}
  \city{Cambridge}
  \state{MA}
  \postcode{02139}
  \country{USA}
}
\email{natsothanaphan@gmail.com}

\begin{abstract}
The DeGroot model of naive social learning assumes that agents only communicate scalar opinions. 
In practice, agents communicate not only their opinions, but their confidence in such opinions. We propose a model that captures this aspect of communication by incorporating signal informativeness into the naive social learning scenario. Our proposed model captures aspects of both Bayesian and naive learning. Agents in our model combine their neighbors' beliefs using Bayes' rule, but the agents naively assume that their neighbors' beliefs are independent. Depending on the initial beliefs, agents in our model may not reach a consensus, but we show that the agents will reach a consensus under mild continuity and boundedness assumptions on initial beliefs. This eventual consensus can be explicitly computed in terms of each agent's centrality and signal informativeness, allowing joint effects to be precisely understood. We apply our theory to adoption of new technology. In contrast to \citet{BBCM}, we show that information about a new technology can be seeded initially in a tightly clustered group without information loss, but only if agents can expressively communicate their beliefs.
\end{abstract}

\maketitle

\section{Introduction}

This paper introduces an intuitive naive social learning model that extends the DeGroot model when agents are differently informed. In standard naive learning, each agent's belief is a scalar, and each agent updates her beliefs by taking the weighted averages of her neighbors' beliefs with fixed weights. This implicitly assumes that every agent's signal quality is the same, which is often not true in practice. We propose a solution to this problem by representing each agent's belief as a probability distribution over possible states of the world.
Signal quality can then be naturally communicated by using Bayes' rule to combine one's belief with the beliefs of one's neighbors.
For example, if an agent has a scalar opinion $\mu$ and a scalar confidence $\tau$ indicating how sure she is of that scalar opinion, her belief can be represented by the Gaussian $\N(\mu,1/\tau)$.
Our model can also accommodate other forms of information.
For example, an agent may believe that either one of the two values is very likely, while other values are less likely. This belief can be represented
as a double-peaked distribution, but it cannot be represented in standard DeGroot models.

It is known separately in Bayesian learning and naive learning literatures that higher quality signals and more centrally located agents, respectively, have more influence in social learning. The innovation in our model is to introduce  signal quality into the naive learning setting, allowing us to precisely investigate how signal quality and agent centrality interact.
The two notions are intertwined in the \textit{weighted likelihood function}, where the formal expression is at equation (\ref{eqn:weightedlik}) in Section \ref{subsec:contribute}. Informally, the value of this function at $\theta$ is the product of the votes of each agent, where an agent's vote is how likely the state value $\theta$ is after the agent received the signal compared to before (the signal quality part), and the number of votes is the agent's eigenvector centrality (the centrality part). 

Our main theorem is that, under some conditions, long-run beliefs are concentrated near the maximizers of the weighted likelihood function. If we know the probability densities of initial beliefs, we can explicitly calculate the consensus belief. This allows us to analyze the interaction; for example, we can calculate how much more informative a signal has to be to offset the centrality disadvantage. This formula can have practical applications in the analysis of information spread in decentralized networks, such as propagation of rumors or adoption of new technologies. Practitioners interested in empirical results can also use our formula without having to fully grasp the theory in the paper. Some examples of applications are discussed in Section \ref{sec:apps}.

\citet{rj} propose essentially the same model in this paper in the case where the state space is finite, and they have derived the same concentration result as ours in that case. Our novel contribution is in the analysis of feasibility and convergence of beliefs when the state space is infinite discrete or continuous. We also draw two policy implications from the model with Gaussian beliefs: how to seed opinion leaders for technology adoption, and how to solve the clustered seeding problem introduced by \citet{BBCM}. We discuss the relationship between our work and \citet{rj} in Subsection \ref{subsec:relatedlit}.


Finally, we contrast the results of our model with that of \citet{BBCM}.
Their model tries to solve the issue of differently informed agents in DeGroot learning by assuming that each agent can either hold a scalar belief or be uninformed. While their model can capture the spread of awareness of new things, it does not have the notion of signal informativeness that our model provides.


Our model also gives a different prediction than that of \citet{BBCM} 
on how to seed a new technology for widespread adoption.
Their model has a \textit{clustered seeding} problem: there is substantial information loss if the initially informed agents (``seeds'') are clustered in a small group of communities. This is a problem because a social planner who wants to maximize the impact of a newly introduced technology should seed it to centrally located opinion leaders, but these leaders are often closely connected to each other. Our model does not have this problem; the influence of each agent depends only on the agent's signal informativeness and centrality. The key modeling difference is that agents in our model can communicate confidence of their beliefs that their neighbors can take into account. Therefore, to solve the clustered seeding problem, the social planner should encourage open and expressive communication between agents.


\subsection{Contributions}
\label{subsec:contribute}

Our agents are in a social network, represented as a strongly connected directed graph.
Assume that there is a true state of the world $\theta^*$. Each agent gets a signal drawn from a distribution that depends on $\theta^*$.
She then forms a \emph{belief}, the posterior probability distribution of $\theta^*$ given the signal.
In each round, each agent has access to beliefs of her neighbors and naively assumes that these all come from \textit{independent} signals. She then forms the Bayesian posterior of $\theta^*$ given the signals of her own and her neighbors as her new belief. This gives rise to an update rule that combines her belief with her neighbors' beliefs that she uses in every round.

We prove that our learning process can be viewed as follows.
At every time step, each agent passes copies of signals she has to her neighbors, so copies of signals ``flow'' through the network. Every copy is taken as independent and no copy is discarded, so the number of copies of agent $j$'s signal that agent $i$ has at round $t$ is the number of paths from $j$ to $i$ of length $t$. Agent $i$'s belief at time $t$ is then the posterior belief conditional on these repeated signals (Prop. \ref{prop:underlyingfeasible}).

We then tackle the important question of what happens to agents' beliefs in the long run. The answer depends on the following \textit{weighted likelihood function} $L$ (Def. \ref{def:weightedlikelihood}).
Let $f_i^{(t)}(\theta)$ be agent $i$'s belief of $\theta^*$ in round $t$ and $f_*(\theta)$ be the common prior of $\theta^*$. 
Then
\begin{align}\label{eqn:weightedlik}
   L(\theta) = \prod_{i=1}^{N} \paren{\frac{f_i^{(0)}(\theta)}{f_*(\theta)}}^{v_i}, 
\end{align}
where $v_i$ is the eigenvector centrality of agent $i$.
We show that if $L(\theta_1)<L(\theta_2)$, then as $t\to\infty$,
\begin{equation}
\label{eq:ratio0}
\frac{f_i^{(t)}(\theta_1)}{f_i^{(t)}(\theta_2)}\to 0
\end{equation}
(Prop. \ref{prop:maxdominate}).
In other words, agents' beliefs that the true parameter is $\theta_2$ are much stronger than the corresponding beliefs for $\theta_1$ as time goes on.
Thus, in the long run, agents' beliefs should concentrate at values $\theta$ that maximize or nearly maximize $L$.

The form of the weighted likelihood function highlights two factors that determine the quality of learning. An agent who is more \textit{centrally located} (high $v_i$) and has a more \textit{informative signal} (high or low $f_i^{(0)}(\theta)/f_*(\theta)$) has more influence over the consensus belief.
Note that our consensus belief
is fundamentally different from that of DeGroot.
In our model, the limit belief is the weighted likelihood maximizer, while in DeGroot, it is a weighted average of beliefs.

The fact that beliefs tend to converge to a point distribution agrees with the following intuition:
agents take a growing number of repeated signals as independent, so they become increasingly confident in their beliefs. 
The fact that the limit belief is tractable and interpretable comes as a surprise, and we believe this to be an attractive feature of our model.

In spite of this intuition and much to our astonishment, the concentration previously discussed turns out not to hold in general if $\Theta$, the state space, is infinite.
We discover that even with a single $\theta_{\max}$ uniquely maximizing $L$, agents' beliefs may not converge to the point distribution at $\theta_{\max}$
(Ex. \ref{ex:countercountable}).
We resolve this difficulty by introducing an extra condition: if $\theta_{\max}$ still maximizes a ``perturbed'' version of $L$, then the concentration result holds.
We can think of this as $\theta_{\max}$ being a sufficiently robust maximizer of $L$.
This result requires careful analysis of convergences of infinite sequences, and the generalization of this idea for an arbitrary state space $\Theta$ can be found in Proposition \ref{prop:lessthanapoint}.

When the state space $\Theta = \R$ is a continuum, for example in the case of Gaussian beliefs, this perturbed $L$ condition is still too strong: perturbing $L$ will almost always shift the maximizer due to the continuity of the space.
The correct condition needs to allow the maximizer to shift, but not too much. Furthermore, it is insufficient to consider the density at a single point, since this tells us nothing about densities at nearby points and so cannot establish concentration. Instead, we need to show that for any neighborhood of $\theta_{\max}$, agents' beliefs that the true state lies in that neighborhood tend to 1.
Our Master Theorem \ref{thm:master} pulls together all of these ideas.

With the Master Theorem \ref{thm:master}, we derive a convergence theorem when the state space is a continuum,
provided that the initial beliefs and priors are sufficiently continuous (Def. \ref{def:piecewise}). Our main result in this most difficult case is Theorem \ref{thm:Rkmain}. The following is a simplified version of this theorem, simplified from Corollary \ref{cor:Rkcontinuous}.

\begin{theorem}
\label{thm:mainsimplified}
Assume that the state space $\Theta = \R^k$.
Suppose that the normalized beliefs $f_i^{(0)}(\theta)/f_*(\theta)$ are bounded continuous functions; the weighted likelihood function $L$ attains the maximum $L_{\max}$ uniquely at the point $\theta_{\max}$; and $L$ decays in the sense that there is a bounded set $S$ such that $\sup_{\theta \notin S} L(\theta) < L_{\max}$.
Then, as $t\to\infty$, beliefs $f_i^{(t)}$ converge to the point distribution at $\theta_{\max}$.
\end{theorem}

This result states that whenever the initial beliefs are continuous functions that are bounded with respect to the prior and they decay sufficiently fast, agents' beliefs converge to the maximizer of the weighted likelihood function if this maximizer is unique.
This result encompasses a wide range of probability distributions in practice.
For example, the theorem applies when beliefs are multivariate Gaussian. Therefore, Theorem \ref{thm:mainsimplified} is widely applicable.

The rest of the paper is organized as follows. We discuss 
related literature in the remaining of this section. Section \ref{sec:model} describes our model. Section \ref{sec:mathprelim} discusses mathematical preliminaries. Section \ref{sec:understand} gives fundamental definitions, interprets the updated beliefs, and analyzes feasibility of initial conditions. Section \ref{sec:concentrate} proves the ``Master Theorem'' that establishes concentration of beliefs near the weighted likelihood maximizers. Section \ref{sec:converge} uses the Master Theorem to prove convergence of beliefs in the finite, infinite discrete, and $\R^k$ cases. Section \ref{sec:apps} gives applications in the cases of binary, Poisson, and Gaussian beliefs.

\subsection{Related Literature}
\label{subsec:relatedlit}

Our paper is closely related to a growing literature on learning and belief formation in social networks \cite{banerjee-herding-1992, bikhchandani-hirshleifer-welch-1992-jpe,pathological-outcomes-observational-learning-smith-sorensen,bala-goyal-2000,gale-kariv-2003-bayesian-learning-in-social-networks,rosenberg-solan-vieille-2009-geb}. Much of this literature focuses on Bayesian learning \cite{celen-kariv-2004-geb,general-framework-rational-learning-networks-mueller-frank,lobel-sadler-information-diffusion-networks,bayesian-learning-acemoglu-dahleh-lobel-ozdaglar} and DeGroot learning \cite{demarzo-vayanos-zwiebel,golub-jackson-wisdom-of-crowds}. Comparable to \citet{golub-jackson-wisdom-of-crowds}, our model shows that centrally located agents, as measured by eigenvector centralities, have more influence on the consensus belief. However, in our model, eigenvector centralities appear as duplicate agents rather than linear weights in the limit belief.

Specifically, our paper belongs to the quasi-Bayesian learning literature.
\citet{eyster-rabin-2010,eyster-rabin-2014} and \citet{rahimian-molavi-jadbabaie-2014} study long-run aggregation of information when agents ignore repetition. \citet{li-tan-local-networks} study learning where each agent updates both her belief on the state of the world and her belief on her neighbors' beliefs.
\citet{jadbabaie-non-bayesian-social-learning} also combine DeGroot learning with Bayesian updates.
Features that distinguish their model from ours are
learning objectives and arrivals of information. Agents in our model receive signals only once at the beginning, rather than at every time step, and Bayesian-update based on their neighbors' beliefs rather than taking weighted averages of distributions. Agents in our model can form and update beliefs even if the true state of the world does not exist, and our limit beliefs are biased by influential agents which do not exist in their model. Lastly, likelihood ratios play a central role in our model; they appear in \citet{rahimian-jadbabaie-loglinear} and \citet{theory-non-bayesian-social-learning} in a different context.

\citet{rj} propose a model where agents use Bayes' rule to come up with a belief update rule that is optimal for the first round and naively continue to use the same rule for subsequent rounds. This naive agent behavior is the same as ours, and resembles the microfoundation of the DeGroot update rule proposed in \citet{demarzo-vayanos-zwiebel}. Agents in their model have utilities; they take actions to maximize their utilities and only observe their neighbors' actions. In contrast, agents in our models do not have utilities, do not take actions, and have direct access to their neighbors' beliefs. When the state space is finite, the action space is a probability simplex over the state space, and each agent's utility is the negative Euclidean distance between her action and the point mass at the true state, each agent's action corresponds to her belief. Under these assumptions, their model becomes equivalent to ours when the state space is finite, and they independently derive what we call the naive Bayesian update rule (Def \ref{def:update}) in their Section 4. They also independently prove the concentration result in this special case when the state space is finite (Thm. \ref{thm:mainfinite}) in their Theorem 4.

Finally, our work is related to \citet{BBCM}. In their model,
the informedness of agents is a binary state, informed individuals hold scalar beliefs, and the learning rule is to take the average belief of informed neighbors.
This contrasts with our model
where informedness is a continuous attribute.
Their work identifies the clustered seeding problem discussed in the introduction, which does not occur in our model.
See further discussion of this problem in Section \ref{subsec:gaussian}.

\section{Model}\label{sec:model}

We discuss the assumptions of our model in this section.

\subsection{Social Network}
Our agents are in a social network, modeled as a directed graph $G$ with $N \geq 2$ nodes, labeled $1,2,\dots,N$ corresponding to $N$ agents.
An edge from $i$ to $j$ is denoted $(i,j)$ and means that there is
information flow from agent $i$ to agent $j$. Assume that $G$ is strongly connected, that is, for every $i$ and $j$ there is a path from $i$ to $j$.
Suppose that there is a self-loop $(i,i)$ at every node,
meaning that information flows from every agent to herself.
If there is an edge $(i,j)$, we say that $i$ is an \emph{in-neighbor} of $j$.
Let $N(i)$ denote the set of in-neighbors of $i$.

\subsection{Beliefs}

Let $\Theta$ be the space of parameters $\theta$, modeled as a measure space.
We model beliefs as \emph{probability density functions} on $\Theta$: nonnegative measurable functions $f:\Theta \to \R_{\geq 0}$ such that the integral
$\int_{\theta \in \Theta} f(\theta)=1.$
These reduce to probability mass functions in a discrete setting, i.e. when $\Theta$ is discrete with the counting measure.
Note that we model beliefs as probability densities, which is a special case of general probability distributions. For example, a point distribution on $\R$ cannot be represented as a density, but most probability distributions in practice are densities.
A \emph{prior} on $\Theta$ is any nonnegative measurable function $f:\Theta \to \R_{\geq 0}$. If the integral $\int_{\theta \in \Theta} f(\theta)$
diverges, we call the prior $f$ \emph{improper}.
Note that our improper priors are not necessarily limits of proper priors,
but simply densities whose integrals diverge and so cannot be normalized.
An example of an improper prior is the uninformative prior, $f(\theta)\equiv 1$,
where $\Theta$ is an infinite discrete set or an unbounded interval of $\R$.

Suppose that each agent has a common (possibly improper) prior  $f_*$ for the true state of the world $\theta^*$ on $\Theta$. Assume that $f_*(\theta)>0$ for all $\theta \in \Theta$.
This assumption is important in order to state the update rule for beliefs (e.g. in Definition \ref{def:update}).

\subsection{Belief Updates}
\label{subsec:modelupdates}
Agents update their beliefs based on the beliefs of their in-neighbors using the following rule, which will be justified shortly.

\begin{definition}
\label{def:update}
The \emph{naive Bayesian updated belief} for agent $i$
with respect to beliefs $f_1,\dots,f_N$ is the belief
$$f(\theta) = \frac{f_*(\theta) \prod_{j \in N(i)}{\paren{f_j(\theta)/f_*(\theta)}}}
{\int_{\theta' \in \Theta} f_*(\theta') \prod_{j \in N(i)}{\paren{f_j(\theta')/f_*(\theta')}}},$$
which is well-defined whenever the integral in the denominator is nonzero and finite.
\end{definition}

It is easily checked that this definition defines a belief. We now provide a justification for this definition. In the justification, we will assume that $f_*$ is a probability density function, but
the formula in Definition \ref{def:update} also makes sense for improper priors.

Consider the following scenario called an \emph{underlying scenario}. Suppose that agent $i$ knows a family of probability density functions $(p_{\theta,i})_{\theta \in\Theta}$, indexed by the parameter $\theta$ and the agent $i$, over a common sample space $\Omega_i$.
We make the following marginalizability assumption, \new{which is important in order to reason consistently about marginal densities}.
For any \new{integers} $n_1,\dots,n_N \geq 0$, 
let $X_{i,1},\dots,X_{i,n_i}$ be \new{random variables} drawn independently from $p_{\theta,i}$.
\new{We thus have a joint probability density for these random variables}
$$p\big( (X_{i,j})_{1\leq i \leq n, 1 \leq j \leq n_i},\theta\big) =p(\theta) \prod_{i=1}^N \prod_{j=1}^{n_i} p(X_{i,j}|\theta).$$
Assume that this distribution can be marginalized over $\theta$ with positive density, that is
$$p\big( (X_{i,j})_{1\leq i \leq n, 1 \leq j \leq n_i}\big)=\int_{\theta \in \Theta} p(\theta) \prod_{i=1}^N \prod_{j=1}^{n_i} p(X_{i,j}|\theta)$$
is nonzero and finite.

Agent $i$ receives a signal $X_i$ drawn from $p_{\theta^*,i}$, \new{the distribution corresponding to the true parameter $\theta^*$},
without knowing $\theta^*$.
She then forms a belief $f_i(\theta)$ for $\theta^*$
conditional on \new{the signal} $X_i$.
Suppose that agent $i$ learns about beliefs $f_j(\theta)$ of all her in-neighbors $j \in N(i)$, \new{formed in the same way}, and wishes
to perform a Bayesian update on her own belief. The result \new{will be} the belief $f(\theta)$
for $\theta^*$ conditional on \new{signals} $X_j$ for every $j \in N(i)$.
\new{Even though the resulting density seems to depend on $X_j$, we will show in the next proposition} that this update can be performed knowing only the beliefs $f_j$ but not the private signals $X_j$,
with the formula as in Definition \ref{def:update}.

\begin{prop}
\label{prop:justifyrule}
In an underlying scenario \new{described above}, the updated belief for agent $i$ is the naive Bayesian updated belief with respect to beliefs $f_1,\dots,f_N$ as in Definition \ref{def:update}.
\end{prop}

\begin{shownto}{abb}
\begin{proofsketch}
Use Bayes' rule to compute the updated belief as
$$f(\theta)=p\paren{\theta|(X_j)_{j \in N(i)}}
= \frac{p(\theta) \prod_{j \in N(i)} p\paren{X_j |\theta} }{\int_{\theta' \in \Theta} p(\theta') \prod_{j \in N(i)} p\paren{X_j |\theta'} }.$$
Another application of Bayes' rule on $p(X_j|\theta)$ finishes the proof.
\end{proofsketch}
\end{shownto}

\begin{shownto}{full}
\begin{proof}
By definition, the updated belief is
$$f(\theta)=p\paren{\theta|(X_j)_{j \in N(i)}}
= \frac{p(\theta)p\paren{(X_j)_{j \in N(i)} |\theta} }{\int_{\theta' \in \Theta} p(\theta') p\paren{(X_j)_{j \in N(i)} |\theta'} }
= \frac{p(\theta) \prod_{j \in N(i)} p\paren{X_j |\theta} }{\int_{\theta' \in \Theta} p(\theta') \prod_{j \in N(i)} p\paren{X_j |\theta'} },$$
where the second equality is Bayes' rule and the last equality holds because $X_j$'s are conditionally independent given any $\theta$.
\new{Note that} the integral in the denominator is nonzero and finite due to the marginalizability assumption.
\new{Because} $p(\theta)=f_*(\theta)$ and
$$p(X_j|\theta) = \frac{p(\theta|X_j)}{p(\theta)} p(X_j) = \frac{f_j(\theta)}{f_*(\theta)}p(X_j),$$
substituting these into our first equation
gives the desired formula.
\end{proof}
\end{shownto}

\subsection{Belief Dynamics}
\new{Here we describe the evolution of beliefs.} Agents share a common prior $f_*$ of $\theta^*$, \new{which may be improper},
and agent $i$ starts off with \emph{initial belief} $f_i^{(0)}$ of $\theta^*$.
\new{The functions} $f_*,f_1^{(0)},\dots,f_N^{(0)}$ are the \emph{initial condition} of the model.
\new{The subtlety is that not all initial conditions lead to well-defined updates. Nevertheless, in Section \ref{subsec:feasible_initial_conditions}, we will provide simple checks that guarantee the existence of updated beliefs.}
For $t=0,1,2,\dots$,
let $f_i^{(t)}$ be agent $i$'s belief of $\theta^*$ at round $t$.
Agents update beliefs as follows. 
For $t\geq 0$, $f_i^{(t+1)}$ is the naive Bayesian updated belief (Def. \ref{def:update}) for agent $i$ with respect to beliefs $f_1^{(t)},\dots,f_N^{(t)}$.

We call an initial condition \emph{feasible} if the updated beliefs are well-defined in every round, that is, the integral in Definition \ref{def:update} is nonzero and finite.
\new{Moreover}, an initial condition \emph{has an underlying scenario} if it arises from an underlying scenario in the \new{above} justification of Definition \ref{def:update}.
Proposition \ref{prop:underlyingfeasible} shows that an initial condition with an underlying scenario is always feasible.

For an initial condition with an underlying scenario, \new{note that although the update rule is rationally justified, it is} only optimal for the first round. In later rounds, an optimal update needs to take into account that the agents' beliefs are already updated at least once and so no longer independent. Still, our naive agents use the same update rule for every round. This modeling choice formalizes what we said earlier that people use Bayes' rule to learn, but they ignore repeated information.

It may seem that using Bayes' rule on the entire distribution requires agents to communicate an infinite amount of information.
Nevertheless, if conjugate priors are used (say, Gaussian beliefs and Gaussian signals, or Beta beliefs and Binomial signals), then agents only need to keep track of a finite number of distribution parameters.
Moreover, our model is robust under \emph{discretization} \new{in the following sense}.
Suppose that agents can only learn about other agents' beliefs at a fixed \new{finite} set of points. By Proposition \ref{prop:ratio_of_f_asymptotic}, the ratio of beliefs $f_i^{(t)}(\theta_1)/f_i^{(t)}(\theta_2)$ is \new{completely} determined by the initial ratios $f_j^{(0)}(\theta_1)/f_j^{(0)}(\theta_2)$ \new{for $1\leq j \leq N$}, so \new{the ratios of beliefs at these fixed points are preserved and hence so are} the approximate shapes of beliefs.
These ratios \new{also} dictate convergence of beliefs towards the maximizer of the weighted likelihood function (Sec. \ref{sec:concentrate}). Therefore, assuming that we choose points close to the maximizer and beliefs are sufficiently continuous, the discretized beliefs will converge to the approximate maximizer.

We name our model \emph{naive Bayesian} for two reasons. First, our model combines features from naive learning and Bayesian learning. Second, our agents assume that distributions of neighbors are independent conditional on the state, which is also the assumption made in Naive Bayes classifiers in machine learning.

\section{Mathematical Preliminaries}
\label{sec:mathprelim}

\new{To analyze our model, we require the following mathematics. We review Perron-Frobenius theory, which helps us determine the behavior of large powers of matrices and establish convergence of beliefs, and the theory of $L^p$ spaces, which helps us bound functions and establish feasibility of initial conditions.}

\subsection{Perron-Frobenius Theory}
\label{subsec:perron}
Let $G$ be the graph of our model, and let $A$ be the adjacency matrix of $G$:
$A_{ij}=1$ if there is an edge $(i,j)$ and 0 otherwise.
The Perron-Frobenius theory provides us with properties of eigenvalues and eigenvectors of $A$.
\begin{shownto}{abb}
We omit the proofs of Lemmas \ref{lem:primitive} and \ref{lem:spectral_radius_greater_than_one}.
\end{shownto}

\begin{definition}
A square matrix $M$ is \emph{primitive} if its entries are nonnegative and the entries of $M^k$ are positive for some $k>0$.
\end{definition}

\new{For an adjacency matrix, a power of the matrix counts the number of paths between each pair of nodes. Hence we would expect that for a strongly connected social network with self-loops, the corresponding adjacency matrix should be primitive. This is formalized in the following lemma.}

\begin{lemma}
\label{lem:primitive}
$A$ is a primitive matrix.
\end{lemma}

\begin{shownto}{full}
\begin{proof}
By strong connectedness, for every $i$ and $j$, there is a path $P_{ij}$ from $i$ to $j$. Let $k$ be the maximum length of all $P_{ij}$'s. We can extend a path to arbitrary lengths by using self-loops, so  there is a path from \new{any $i$ to any $j$} of length $k$. Because the $(i,j)$ entry of $A^k$ is the number of paths from $i$ to $j$ of length $k$ (see e.g. Lemma \ref{lem:path_recurrence}), $A^k$ has positive entries. Therefore $A$ is primitive.
\end{proof}
\end{shownto}

The following theorem from \citet[Chapter 8]{meyer}, \new{a version of Perron-Frobenius for primitive matrices}, characterizes the largest eigenvalue of a primitive matrix and the growth rate of the matrix's powers.

\begin{theorem}[Perron-Frobenius for Primitive Matrices]
\label{thm:perron}
Let $M$ be a primitive matrix. Let $r=\rho(M)$ be the spectral radius of $M$, that is, the maximum absolute value of eigenvalues of $M$. Then $r$ is an eigenvalue of $M$, any other eigenvalue of $M$ has absolute value less than $r$, and there are unique vectors $v$ and $w$ with positive entries such that
$$Mv=rv,\quad wM=rw,\quad \norm{v}_1=\norm{w}_1=1.$$
\new{Note that $\norm{\cdot}_1$ denotes the sum of entries.}
Call $v$ the \emph{Perron vector} of $M$. \new{Then} $w$ is the Perron vector of $M^\top$.

\new{Finally}, the limit $\lim_{k \to \infty} M^k/r^k$ exists and is given by
$$P:= \lim_{k \to \infty} \frac{M^k}{r^k} = \frac{vw^\top}{w^\top v}.$$
\new{This limit $P$ is called} the \emph{Perron projection} of $M$.
\end{theorem}

\new{The following lemma will be useful when we reason that the entries of $A^k$ tend to infinity.}

\begin{lemma}\label{lem:spectral_radius_greater_than_one}
The spectral radius of $A$ is greater than 1, that is, $\rho(A)>1$.
\end{lemma}

\begin{shownto}{full}
\begin{proof}
Let $r=\rho(A)$. Lemma \ref{lem:primitive} implies that $A$ is primitive, so by Theorem \ref{thm:perron}, $r$ is an eigenvalue of $A$,
and there is an eigenvector $v$ associated with $r$ with positive entries. Pick an arbitrary edge $(i,j)$ with $i \neq j$. Then
$$rv_i=Av_i=\sum_{i'=1}^N A_{ii'}v_{i'} \geq v_i + v_j > v_i,$$
implying that $r>1$.
\end{proof}
\end{shownto}

\new{We now define the notion of centrality alluded to in the introduction.}

\begin{definition}
If $v$ is the Perron vector of $A$, then call $v_i$ the \emph{eigenvector centrality} of agent $i$.
\end{definition}

Eigenvector centrality is a standard notion that measures the influences of agents. To justify this, note that by definition,
$$rv_i=\sum_{j=1}^N A_{ij}v_j=\sum_{i \in N(j)} v_j.$$
Call $v_i$ the score of agent $i$. \new{Then this equation says that} one's score is large if one's information can reach many agents; \new{moreover}, reaching agents with a larger score contributes more to one's score. This exact behavior characterizes influence: one's influence is large if one talks to many people, and talking to more influential people contributes more to one's influence than talking to less influential people. So it is appropriate to say that eigenvector centrality is a measure of influence.

\subsection{$L^p$ Spaces}
Let $f:\Theta \to \R_{\geq 0}$ be a \new{nonnegative} measurable function. For any $p \in [1,\infty)$, recall that $f \in L^p$ means that the $p$th power of $f$ has finite integral, $\int_{\theta \in \Theta} f^p < \infty$.
Moreover, $f \in L^\infty$ means that $f$ is bounded almost everywhere,
$\esssup_{\theta \in \Theta} f<\infty$.
The following inequality, \new{H\"older's inequality}, is fundamental in the study of $L^p$ spaces. \new{Note that the version presented here has a different parametrization of exponents from the usual formulation.}

\begin{theorem}[H\"older's Inequality]
\label{thm:holder}
Let $f_1,\dots,f_N:\Theta \to \R_{\geq 0}$ be measurable functions.
Let $p_i \in [0,1]$ be such that $p_1+\dots+p_N=1$. Then
$$\int_{\theta \in \Theta} f_1^{p_1}\dots f_N^{p_N} \leq \paren{\int_{\theta \in \Theta} f_1}^{p_1}\dots
\paren{\int_{\theta \in \Theta} f_N}^{p_N}.$$
In particular, if $f_i \in L^1$ for all $i$, then the weighted geometric mean $f_1^{p_1}\dots f_N^{p_N} \in L^1$.
\end{theorem}

\new{We also consider the following ``weighted'' $L^p$ spaces, where the weight is the common prior, which} will be useful in the study of initial conditions in Section \ref{subsec:feasible_initial_conditions}. 

\begin{definition}
For $p \in [1,\infty]$, define the \emph{$L^p$ space weighted by the prior}, $L^p_*$, as follows. If $p \in [1,\infty)$, the space consists of measurable functions $f:\Theta \to \R_{\geq 0}$ such that
$\int_{\theta \in \Theta} f_*f^p<\infty$.
Equip these functions with the norm
$$\norm{f}_{p,*} := \paren{\int_{\theta \in \Theta} f_*f^p}^{1/p}.$$
\new{This is just like the $L^p$ norm but weighted by the prior.} The space $L^\infty_*$ is the same as $L^\infty$ and consists of measurable functions $f:\Theta \to \R_{\geq 0}$ such that $\esssup_{\theta \in \Theta} f < \infty$.
Equip these functions with the norm
$$\norm{f}_{\infty,*}=\esssup_{\theta \in \Theta} f.$$
\new{This is exactly the same as the $L^\infty$ norm.}
\end{definition}

We now present a version of H\"older's inequality on these weighted $L^p$ spaces, \new{as well as} an interpolation result.
\begin{shownto}{abb}Lemma \ref{lem:newholder} follows from Theorem \ref{thm:holder} and its proof is omitted.\end{shownto}

\begin{lemma}
\label{lem:newholder}
Let $p \in [1,\infty)$ and let $f_1,\dots,f_N \in L^p_*$. Let $p_i \in [0,1]$ be such that $p_1+\dots+p_N=1$.
Then $f_1^{p_1}\dots f_N^{p_N}\in L^p_*$, with
$$\norm{f_1^{p_1}\dots f_N^{p_N}}_{p,*}
\leq \norm{f_1}_{p,*}^{p_1}\dots \norm{f_N}_{p,*}^{p_N}.$$
\end{lemma}

\begin{shownto}{full}
\begin{proof}
Apply H\"older's inequality (Theorem \ref{thm:holder}) to the functions $f_*f_i^p$.
\end{proof}
\end{shownto}

\begin{lemma}\label{lem:interpolationLp}
If $1 \leq p < q < r \leq \infty$, then $L^p_* \cap L^r_* \subseteq L^q_*$.
\end{lemma}

\begin{shownto}{abb}
\begin{proofsketch}
The case $r < \infty$ follows from Lemma \ref{lem:newholder}. Prove the case $r=\infty$ directly.
\end{proofsketch}
\end{shownto}

\begin{shownto}{full}
\begin{proof}
First assume that $r < \infty$. \new{Then} there are $p_1,p_2 \in (0,1)$ such that $p_1+p_2=1$ and $p_1p+p_2r=q$. If $f \in L^p_* \cap L^r_*$, then $f^p,f^r \in L^1_*$.
By Lemma \ref{lem:newholder}, $f^{p_1p+p_2r}=f^q \in L^1_*$, so that $f \in L^q_*$.

We directly prove the case where $r=\infty$. Let $M=\norm{f}_{\infty,*}$. Then $f \leq M$ almost everywhere, so
$$\int_{\theta \in \Theta} f_*f^q \leq
M^{q-p}\int_{\theta \in \Theta} f_*f^p<\infty.$$
Therefore $f \in L^q_*$.
\end{proof}
\end{shownto}

\section{Understanding the Model}\label{sec:understand}

\new{In Section \ref{subsec:prelimdef}, we introduce some preliminary definitions that will be useful in understanding the model in Sections \ref{subsec:interpret} and  \ref{subsec:feasible_initial_conditions}.
Then we interpret the updated beliefs in Section \ref{subsec:interpret}. Finally, we find the conditions that guarantee that an initial condition is feasible in Section \ref{subsec:feasible_initial_conditions}.}

\subsection{Preliminary Definitions}
\label{subsec:prelimdef}
\new{We introduce the notion of \emph{normalized beliefs}.}

\begin{definition}\label{def:normalized_belief}
The \emph{normalized belief} of agent $i$ at round $t$ is the function 
$$g_i^{(t)}(\theta) := \frac{f_i^{(t)}(\theta)}{f_*(\theta)}.$$
\end{definition}

\new{We can interpret this as how large a a belief is compared to the prior. This quantity will appear in the weighted likelihood function in Definition \ref{def:weightedlikelihood} and play a central role in our analysis.}
Note that the normalized belief is well-defined because $f_*(\theta)>0$ by assumption, and it has the property that it is in $L^1_*$ with norm 1:
$\big\lVert g_i^{(t)}\big\rVert_{1,*}=1$.
However, it need not be a belief, as the integral $\int_{\theta \in \Theta} g_i^{(t)}(\theta)$ may diverge.

With normalized beliefs, the update rule in Definition \ref{def:update} assumes the simpler form

\begin{equation}
\label{eq:updateg}
g_i^{(t+1)}(\theta) = \frac{\prod_{j \in N(i)}{g_j^{(t)}(\theta)}}
{\int_{\theta' \in \Theta} f_*(\theta') \prod_{j \in N(i)}{g_j^{(t)}(\theta')}}.
\end{equation}

\new{We now introduce the notation for the number of paths between each pair of nodes.}

\begin{definition}
For nodes $i$ and $j$ of $G$ and $t\geq 0$, let $P_{ij}^{(t)}$ be the number of paths from $i$ to $j$ of length $t$.
By convention, there is one path of length 0 from a node to itself and no path of length 0 from a node to a different node.
\end{definition}

The following simple combinatorial lemma characterizes the number of paths.
\new{It is completely standard but we include it here for completeness.}
\begin{shownto}{abb}We omit its proof.\end{shownto}

\begin{lemma}\label{lem:path_recurrence}
The numbers of paths satisfy the recurrence $P_{ij}^{(t+1)}=\sum_{k \in N(j)} P_{ik}^{(t)}$.
Moreover, $P_{ij}^{(t)}$ is the $(i,j)$ entry of the matrix $A^t$.
\end{lemma}

\begin{shownto}{full}
\begin{proof}
The recurrence follows by considering that a path from $i$ to $j$ \new{of length $t+1$ can be broken down into paths from $i$ to $k$ of length $t$ and from $k$ to $j$ of length 1 for some $k$}. Let $P^{(t)}$ be the matrix whose $(i,j)$ entry is $P_{ij}^{(t)}$. The recurrence can be written in the form $P^{(t+1)}=P^{(t)}A$, with the initial condition $P^{(0)}=I$. It follows by induction that $P^{(t)}=A^t$.
\end{proof}
\end{shownto}

\new{Finally, we formally define the notion of \emph{$k$-feasibility} and \emph{feasibility}.}

\begin{definition}
For any $k\geq 1$, an initial condition $f_*,f_1,\dots,f_N$ is \emph{$k$-feasible} if the updated beliefs are well-defined up to at least round $k$.
An initial condition is \emph{feasible} if it is $k$-feasible for every $k$.
\end{definition}

\subsection{Interpretation of Updated Beliefs}
\label{subsec:interpret}
We interpret what agents are doing in our model. Our goal is to show that each agent aggregates information where another agent's initial belief is weighted by the number of paths from that agent to the aggregating agent. \new{This means that an agent who has many paths to other agents, i.e. an agent with high centrality, will be influential in our model.} This result will lead to conditions for feasible initial conditions in Section \ref{subsec:feasible_initial_conditions} and show concentration of beliefs in Section \ref{sec:concentrate}.

\new{To carry out this program, we first} give an explicit condition for $k$-feasibility and an explicit formula for the updated beliefs \new{in the next proposition}.

\begin{prop}
\label{prop:update_t_g}
An initial condition $f_*,f_1^{(0)},\dots,f_N^{(0)}$ is $k$-feasible if and only if for all $i$ and $0\leq t \leq k$, the function
$$(g_1^{(0)})^{P_{1i}^{(t)}}\dots (g_N^{(0)})^{P_{Ni}^{(t)}} \in L^1_*$$
with nonzero norm.
In that case, for all $0\leq t \leq k$,
$$g_i^{(t)}(\theta) = \frac{\prod_{j=1}^N {g_j^{(0)}(\theta)^{P_{ji}^{(t)}}}}
{\int_{\theta' \in \Theta} f_*(\theta') \prod_{j=1}^N g_j^{(0)}(\theta')^{P_{ji}^{(t)}}}.$$
\end{prop}

\begin{shownto}{abb}
\begin{proofsketch}
Induct on $k$ and use Lemma \ref{lem:path_recurrence} to get exponents into the form $P_{ji}^{(t)}$.
\end{proofsketch}

\end{shownto}

\begin{shownto}{full}

\begin{proof}

\new{The proof is essentially by induction on $k$.} Any initial condition is certainly 0-feasible. And for any $i$, $\prod_{j=1}^N (g_j^{(0)})^{P_{ji}^{(0)}} = g_i^{(0)} \in L^1_*$ with norm 1. This proves the base case $k=0$.

Assume that the proposition holds for $k$. For a $(k+1)$-feasible initial condition, the updated beliefs given by equation \eqref{eq:updateg}
$$g_i^{(k+1)}(\theta) = \frac{\prod_{j \in N(i)}{g_j^{(k)}(\theta)}}
{\int_{\theta' \in \Theta} f_*(\theta') \prod_{j \in N(i)}{g_j^{(k)}(\theta')}}$$
are well-defined.
By the induction hypothesis, this equals
\begin{equation}
\label{eq:updateformula}
g_i^{(k+1)}(\theta) =
\frac{\prod_{j \in N(i)} \prod_{\ell=1}^{N} g_\ell^{(0)}(\theta)^{P_{\ell j}^{(k)}} }{\int_{\theta' \in \Theta} f_*(\theta') \prod_{j \in N(i)} \prod_{\ell=1}^{N} g_\ell^{(0)}(\theta')^{P_{\ell j}^{(k)}} } = 
\frac{\prod_{\ell=1}^{N} g_\ell^{(0)}(\theta)^{P_{\ell i}^{(k+1)}} }{\int_{\theta' \in \Theta} f_*(\theta') \prod_{\ell=1}^{N} g_\ell^{(0)}(\theta')^{P_{\ell i}^{(k+1)}} },
\end{equation}
where the second equality holds because the exponent of $g_\ell^{(0)}$ is $\sum_{j \in N(i)} P_{\ell j}^{(k)}=P_{\ell i}^{(k+1)}$ by Lemma \ref{lem:path_recurrence}.
So the formula holds. \new{Finally}, the denominator of the right-hand expression shows that $\prod_{j=1}^N (g_j^{(0)})^{P_{ji}^{(k+1)}} \in L^1_*$ with nonzero norm.

Conversely, suppose that $\prod_{j=1}^N (g_j^{(0)})^{P_{ji}^{(t)}} \in L^1_*$ with nonzero norm for all $t\leq k+1$. By the induction hypothesis, the initial condition is $k$-feasible, \new{and we must show that it is $(k+1)$-feasible}. The right-hand expression of equation \eqref{eq:updateformula} is well-defined, so \new{by running the same argument as above backwards}, it can be written as the expression for $g_i^{(k+1)}(\theta)$ in the update rule. Therefore $g_i^{(k+1)}$ is well-defined and the initial condition is $(k+1)$-feasible.
\end{proof}

\end{shownto}

\new{We can now prove that an initial condition with an underlying scenario (see Section \ref{subsec:modelupdates}) is feasible, as a result of the next proposition.}

\begin{prop}
\label{prop:underlyinggood}
Assume that the initial condition $f_*,f_1^{(0)},\dots,f_N^{(0)}$ has an underlying scenario.
For any integers $n_1,\dots,n_N \geq 0$, the function
$$(g_1^{(0)})^{n_1}\dots (g_N^{(0)})^{n_N} \in L^1_*$$
with nonzero norm.
\end{prop}

\begin{shownto}{abb}
\begin{proofsketch}
Follows from the marginalizability assumption: $p\big( (X_i^{n_i})_{i=1}^{N} \big)$ is nonzero and finite, where $(X_i^{n_i})_{i=1}^{N}$ is the ordered tuple of $n_i$ copies of $X_i$ for $1\leq i \leq N$.
\end{proofsketch}
\end{shownto}

\begin{shownto}{full}
\begin{proof}
Recall that agent $i$'s initial belief is $f_i^{(0)}(\theta)=p(\theta|X_i)$, where $X_i$ is her signal.
Denote by $(X_i^{n_i})_{i=1}^{N}$ the ordered tuple of $n_i$ \new{independent} copies of $X_i$ for $1\leq i \leq N$.
By the marginalizability assumption, \new{the joint density of this tuple is}
\begin{align*}
p\big( (X_i^{n_i})_{i=1}^{N} \big) &= \int_{\theta \in \Theta} f_*(\theta) p\big( (X_i^{n_i})_{i=1}^{N} | \theta\big) = \int_{\theta \in \Theta} f_*(\theta) \prod_{i=1}^{N} p\paren{X_i| \theta}^{n_i} \\
&= \int_{\theta \in \Theta} f_*(\theta) \prod_{i=1}^{N} \bigg(\frac{f_i^{(0)}(\theta) p(X_i)  }{f_*(\theta)}\bigg)^{n_i} =
\bigg( \prod_{i=1}^{N} p(X_i)^{n_i}\bigg)
\int_{\theta \in \Theta} f_* \prod_{i=1}^{N} (g_i^{(0)})^{n_i},
\end{align*}
and it is nonzero and finite. 
So $\prod_{i=1}^{N} (g_i^{(0)})^{n_i} \in L_*^{1}$ with nonzero norm.
\end{proof}
\end{shownto}

\new{As a corollary of the previous two propositions, we will show that an initial condition with an underlying scenario is feasible, as well as give an explicit formula for the updated beliefs. This is important as it gives us insight into the update process.}

\begin{prop}
\label{prop:underlyingfeasible}
An initial condition $f_*,f_1^{(0)},\dots,f_N^{(0)}$ with an underlying scenario is feasible. If agent $i$ receives signal $X_i$ in the underlying scenario and $I_i^{(t)}$ is the ordered tuple consisting of $P_{ji}^{(t)}$ \new{independent} copies of $X_j$ for each $j$, then
$$f_i^{(t)}(\theta) = p\big(\theta|I_i^{(t)}\big).$$
\new{In other words, the belief of agent $i$ in round $t$ for each parameter $\theta$ is the belief for that parameter conditional on the information $I_i^{(t)}$.}
\end{prop}

\begin{shownto}{abb}
\begin{proofsketch}
Compute $p\big(\theta|I_i^{(t)}\big)$ as in Proposition \ref{prop:justifyrule} and apply Proposition \ref{prop:update_t_g}.
\end{proofsketch}
\end{shownto}

\begin{shownto}{full}
\begin{proof}
Propositions \ref{prop:update_t_g} and \ref{prop:underlyinggood} \new{directly} imply the first statement. \new{For the second statement}, as in Proposition \ref{prop:justifyrule}, we can compute
\begin{align*}\label{eqn:updated_belief_repeated_signals}
p\big(\theta|I_i^{(t)}\big) =
p\bigg( \theta \Big| \Big( X_j^{P_{ji}^{(t)}} \Big)_{j=1}^N \bigg)
= \frac{f_*(\theta) \prod_{j=1}^{N} g_j^{(0)}(\theta)^{P_{ji}^{(t)}} }{ 
\int_{\theta' \in \Theta} f_*(\theta') \prod_{j=1}^{N} g_j^{(0)}(\theta')^{P_{ji}^{(t)}}
    }.
\end{align*}
\new{By Proposition \ref{prop:update_t_g}, this quantity} equals $f_i^{(t)}(\theta)$.
\end{proof}
\end{shownto}

Proposition \ref{prop:underlyingfeasible} gives us an intuition of the process. At every time step, each agent passes copies of signals she has to her neighbors, so copies of signals ``flow'' through the network. Every copy is taken as independent and no copy is discarded, so the number of copies of agent $j$'s signal that agent $i$ has at round $t$ is the number of paths from $j$ to $i$ of length $t$.
Moreover, we can interpret the normalized belief $g_i^{(t)}(\theta)$ as the amount of support that $I_i^{(t)}$ provides for $\theta$.
Since $g_i^{(t)}(\theta) = p(\theta|I_i^{(t)})/p(\theta)$ by Proposition \ref{prop:underlyingfeasible}, this quantity is larger or smaller than 1 if learning $I_i^{(t)}$ increases or decreases the probability density at $\theta$, respectively.

From Proposition \ref{prop:underlyingfeasible}, we can qualitatively predict how our model behaves in the long term.
Because repeated signals in $I_i^{(t)}$ are taken as independent, agents should become overconfident over time. Furthermore, because the numbers of repeated signals are the numbers of paths, agents with more paths to other agents (those with higher eigenvector centralities) should be more influential. Our results on concentration and convergence of beliefs in Sections \ref{sec:concentrate} and \ref{sec:converge} confirm these predictions.

\subsection{Initial Conditions}
\label{subsec:feasible_initial_conditions}

\new{We now build on the work of the previous section to} investigate containment relationships between classes of initial conditions and present simple sufficient conditions for feasibility.

\begin{definition}
An initial condition $f_*,f_1,\dots,f_N$ is \emph{proper} if $\int_{\theta \in \Theta} f_*<\infty$.
It is \emph{improper} if $\int_{\theta \in \Theta} f_*=\infty$.
Let $\mc{P}$ and $\mc{IP}$ denote classes of proper and improper initial conditions.

Let $\mc{F}$ and $\mc{F}_k$ denote classes of feasible and $k$-feasible initial conditions. Let $\mc{U}$ denote the class of initial conditions with underlying scenarios.
\end{definition}

\new{The notion of properness may be of interest because it tells us whether the prior can be realized as a probability density. On the one hand, a prior which is a probability density seems to correspond to real-world situations, but on the other hand, the ``uninformative prior'' $f_* \equiv 1$ is a natural choice in many situations but is often improper.
Note that an initial condition with an underlying scenario is proper by definition, i.e. $\mc{U} \subseteq \mc{P}$.}

\begin{figure}[h]
\begin{center}
\begin{tikzpicture}
    \draw (0,0) rectangle (4,4);
    \draw (0,4) node[anchor=north west] {$\mathcal{P}$};
    \draw (4,0) rectangle (8,4) node[anchor=north east] {$\mathcal{IP}$};
    \draw (3.4,2) circle (0.5);
    \node at (3.4,2) {$\mathcal{U}$};
    \draw (4,2) ellipse (1.6 and 0.8);
    \node at (5.2,2) {$\mathcal{F}$};
    \draw[dashed] (4,2) ellipse (2.5 and 1.25);
    \node at (6.05,2.25) {$\mathcal{F}_{k+1}$};
    \draw[dashed] (4,2) ellipse (3.3 and 1.65);
    \node at (6.7,2.55) {$\mathcal{F}_k$};
\end{tikzpicture}
\end{center}
\caption{Containment relationships between classes of initial conditions, where \new{in the figure} $k$ can be any positive integer. There is a parameter space $\Theta$ and a graph $G$ such that all regions in the diagram are nonempty \new{for all values of $k$} (Prop. \ref{prop:diagramworks}).}
\label{fig:diagram}
\end{figure}
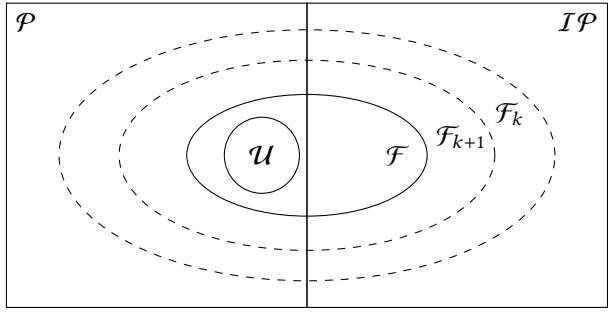

It is obvious that
$\mc{F}_{k+1} \subseteq \mc{F}_k$,
$\mc{F} \subseteq \mc{F}_k$ and
$\mc{U} \subseteq \mc{P}$.
Proposition \ref{prop:underlyingfeasible} shows \new{the nontrivial containment} $\mc{U} \subseteq \mc{F}$.
\new{Thus, the classes are contained in one another} as shown in the diagram in Figure \ref{fig:diagram}.
We now show that there is $\Theta$ and $G$ such that all regions in the diagram are nonempty. In particular, these containments are strict, \new{and so the classes do not collapse into one another. For example, it is not sufficient to show that an initial condition is feasible by checking 1-feasibility.}

\begin{prop}
\label{prop:diagramworks}
Let $\Theta = \{1,2,3,\dots\}$ and $G$ be the graph with 2 nodes and all 4 possible edges. Then all regions in the diagram in Figure \ref{fig:diagram} are nonempty \new{for all values of $k$}.
\end{prop}

\begin{shownto}{abb}
\begin{proofsketch}
See the full version of the paper for an explicit example in each region.
\end{proofsketch}
\end{shownto}

\begin{shownto}{full}

\begin{proof}
\new{An underlying scenario clearly exists (e.g. take $f_*(\theta) = c\theta^{-2}$ for an appropriate $c$ and let agents always receive a constant signal 0: $\Omega_i = \set{0}$ and $p_{\theta,i}(0)=1$), so} $\mc{U}\neq \varnothing$. Examples of initial conditions in other regions are given in the following table, where $c,c_1,c_2$ are normalizing constants.

\begin{center}
\begin{tabular}{| c | c | c | c | c |}
\hline
Region & $f_*(\theta)$ & $f_1^{(0)}(\theta)$ & $f_2^{(0)}(\theta)$ \\ \hline
$\mc{P} \setminus \mc{F}_1$ & $\theta^{-3}$ & $c\theta^{-2}$ & $c\theta^{-2}$ \\ \hline
$\mc{IP} \setminus \mc{F}_1$ & $\theta^{-3}$ for odd $\theta$, $1$ for even $\theta$ & $c\theta^{-2}$ & $c\theta^{-2}$ \\ \hline
$\paren{\mc{F}_k \cap \mc{P}} \setminus \paren{\mc{F}_{k+1}\cap \mc{P}}$ & $\theta^{-(2^{k+1}+1)}$ & $c\theta^{-2^{k+1}}$ & $c\theta^{-2^{k+1}}$ \\
\hline
$\paren{\mc{F}_k \cap \mc{IP}} \setminus \paren{\mc{F}_{k+1} \cap \mc{IP}}$ & $\theta^{-(2^{k+1}+1)}$ for odd $\theta$, $1$ for even $\theta$ & $c\theta^{-2^{k+1}}$ & $c\theta^{-2^{k+1}}$ \\ \hline
$\paren{\mc{F}\cap\mc{P}} \setminus \mc{U}$ & $\theta^{-3}$ & $c_1 \theta^{-2}$ & $c_2 \theta^{-4}$ \\ \hline
$\mc{F}\cap \mc{IP}$ & $1$ & $c\theta^{-2}$ & $c\theta^{-2}$ \\ \hline
\end{tabular}
\end{center}

\new{The following considerations will help verify the table.} An initial condition is in $\mc{P}$ if $\int_{\theta} f_*<\infty$ and is in $\mc{IP}$ otherwise.
For any $t\geq 1$, there are $2^{t-1}$ paths from any node $i$ to any node $j$. By Proposition \ref{prop:update_t_g}, an initial condition is in $\mc{F}_k$ if and only if
$$ \int_{\theta} f_* (g_1^{(0)})^{2^{t-1}} (g_2^{(0)})^{2^{t-1}}$$
is nonzero and finite for all $1 \leq t\leq k$, and it is in $\mc{F}$ if and only if this is true for all $t \geq 1$.
These checks are straightforward using the fact that $\sum_{\theta=1}^{\infty} \theta^{-p}$ is finite for $p>1$ and infinite for $p \leq 1$. 
To check that the example for $\paren{\mc{F}\cap\mc{P}} \setminus \mc{U}$ does not have an underlying scenario, use Proposition \ref{prop:underlyinggood} and the fact that
$$\int_{\theta} f_* (g_1^{(0)})^2 = c_1^2 \int_{\theta} \theta^{-1}=\infty.$$
\end{proof}

\end{shownto}

Now we \new{develop} a sufficient condition for feasibility and show that it is \new{``individually optimal,'' in the sense that will be explained in Proposition \ref{prop:suff_feas_weakest}}. We first take care of the degeneracy issue.

\begin{definition}
An initial condition $f_*,f_1^{(0)},\dots,f_N^{(0)}$ is \emph{nondegenerate} if the set
$\{\theta \in \Theta: f_i^{(0)}(\theta)>0 \text{ for all } i\}$ does not have measure zero. It is \emph{degenerate} otherwise.
\end{definition}

\new{Alternatively, an initial condition is degenerate if the product of initial beliefs vanishes almost everywhere. We might expect that degeneracy is a problem for performing updates as the denominator in Definition \ref{def:update} might be zero. The following proposition makes this rigorous.}

\begin{prop}
\label{prop:nondeg}
A feasible initial condition is nondegenerate.
\end{prop}

\begin{shownto}{abb}
\begin{proofsketch}
For large $k$, use the nonzero norm condition in Proposition \ref{prop:update_t_g}.
\end{proofsketch}
\end{shownto}

\begin{shownto}{full}
\begin{proof}
Pick any node $i$ and pick $k$ such that there is a path from any $j$ to $i$ of length $k$. Let $n_j=P_{ji}^{(k)}>0$. Then Proposition \ref{prop:update_t_g} implies that
$(g_1^{(0)})^{n_1}\dots (g_N^{(0)})^{n_N} \in L^1_*$
with nonzero norm. If the initial condition is degenerate, then this function vanishes almost everywhere and must have zero norm, \new{contradiction}. Therefore the initial condition is nondegenerate.
\end{proof}
\end{shownto}

\new{Conversely, by assuming nondegeneracy, we can prove feasibility by requiring some form of boundedness. Specifically, it suffices for the normalized beliefs to be in $L^p_*$ for $p$ large; this is stated precisely below.}

\begin{prop}
\label{prop:sufficient_feasibility}
A nondegenerate initial condition $f_*,f_1^{(0)},\dots,f_N^{(0)}$ is feasible if for every $i$ and $M<\infty$, there is $M \leq p \leq \infty$ such that $g_i^{(0)} \in L^p_*$.
\end{prop}

\begin{shownto}{abb}
\begin{proofsketch}
Since $g_i^{(0)} \in L_*^1$, we can interpolate by Lemma \ref{lem:interpolationLp} to get $g_i^{(0)} \in L_*^p$ for all $1\leq p < \infty$.
Lemma \ref{lem:newholder} implies that the feasibility condition in Proposition \ref{prop:update_t_g} is satisfied.
\end{proofsketch}
\end{shownto}

\begin{shownto}{full}
\begin{proof}
By Proposition \ref{prop:update_t_g}, it suffices to prove that for any $n_1,\dots,n_N\geq 0$ that do not vanish simultaneously,
$\prod_{i=1}^{N} (g_i^{(0)})^{n_i} \in L_*^{1}$ with nonzero norm.
The ``nonzero norm'' part follows from nondegeneracy: $\{\theta \in \Theta: g_i^{(0)}(\theta)>0 \text{ for all } i\}$ does not have measure zero. \new{To prove that it is in $L_*^{1}$, we proceed by way of interpolation (Lemma \ref{lem:interpolationLp}) and H\"older (Lemma \ref{lem:newholder}).}
Let $M=n_1+\dots+n_N \geq 1$.
\new{By hypothesis}, for each $i$, there is $M \leq p\leq \infty$ such that 
$g_i^{(0)} \in L_*^p$. Because $g_i^{(0)} \in L_*^1$ by definition, by Lemma \ref{lem:interpolationLp}, $g_i^{(0)}\in L^M_*$.
Then by Lemma \ref{lem:newholder}, $\prod_{i=1}^{N} (g_i^{(0)})^{n_i/M} \in L_*^M$.
Hence $\prod_{i=1}^{N} (g_i^{(0)})^{n_i} \in L_*^1.$
\end{proof}
\end{shownto}

\new{One might ask whether the requirement in the previous proposition can be further weakened. Unfortunately,} the following proposition shows that the hypothesis $g_i^{(0)} \in L^p_*$ in Proposition \ref{prop:sufficient_feasibility} is already the \new{optimal} one to put on each $g_i^{(0)}$ individually.

\begin{prop}\label{prop:suff_feas_weakest}
Fix $f_*$. Let $S \subseteq L_*^1$ be a set with the following property: any nondegenerate initial condition $f_*,f_1^{(0)},\dots,f_N^{(0)}$ such that $g_i^{(0)} \in S$ for all $i$ is feasible.
Then $S \subseteq L_*^{p}$ for all $1\leq p <\infty$.
\end{prop}

\begin{shownto}{abb}
\begin{proofsketch}
For any $f \in S$, consider an initial condition with $g_i^{(0)}=f$.
By Proposition \ref{prop:update_t_g},
$f^{P_{1i}^{(t)}+\dots+P_{Ni}^{(t)}} \in L^1_*$.
The exponent goes to infinity by Theorem \ref{thm:perron} and  Lemmas \ref{lem:path_recurrence} and \ref{lem:spectral_radius_greater_than_one}, so we can interpolate using Lemma \ref{lem:interpolationLp}.
\end{proofsketch}
\end{shownto}

\begin{shownto}{full}
\begin{proof}
Pick any $f \in S$ and $1 \leq p <\infty$. We must show that $f \in L_*^p$. If $f$ vanishes almost everywhere, then clearly $f \in L_*^p$. So suppose otherwise. An initial condition with $g_i^{(0)}=f$ is well-defined \new{(i.e. $f_i^{(0)}$ has integral 1)} because $f \in L^1_*$, and it is nondegenerate, so it is feasible by hypothesis.
\new{We show that this feasibility forces $f$ to be in $L^p_*$.}
By Proposition \ref{prop:update_t_g}, the function
$f^{P_{1i}^{(t)}+\dots+P_{Ni}^{(t)}} \in L^1_*$ for any $i$ and $t$.
By Lemma \ref{lem:path_recurrence}, these $P_{ji}^{(t)}$'s are entries of $A^t$, which go to infinity as $t\to\infty$ by Theorem \ref{thm:perron} and Lemma \ref{lem:spectral_radius_greater_than_one}.
So we obtain a sequence $p_t \to \infty$ as $t\to\infty$ such that $f \in L^{p_t}_*$. By Lemma \ref{lem:interpolationLp}, \new{we can interpolate to get} $f \in L^p_*$.
\end{proof}
\end{shownto}

\new{Finally,} Proposition \ref{prop:sufficient_feasibility} can be slightly weakened to give a sufficient condition for feasibility that is easy to check in practice \new{in the next corollary.
This weaker version  is used to establish feasibility in e.g. Corollary  \ref{cor:Rkcontinuous}
and Theorem \ref{thm:mainsimplified}.}
\begin{shownto}{abb}We omit its proof.\end{shownto}

\begin{cor}
\label{cor:suffcientfeas}
A nondegenerate initial condition $f_*,f_1^{(0)},\dots,f_N^{(0)}$ is feasible if for every $i$, $g_i^{(0)} \in L^{\infty}$, that is, $g_i^{(0)}$ is bounded almost everywhere.
\end{cor}

\begin{shownto}{full}
\begin{proof}
By definition, $g_i^{(0)} \in L_*^{1}$, so by Lemma \ref{lem:interpolationLp}, $g_i^{(0)} \in L_*^{p}$ for any $1 \leq p < \infty$. \new{Thus the requirements of} Proposition \ref{prop:sufficient_feasibility} \new{are satisfied}.
\end{proof}
\end{shownto}

\section{Concentration of Beliefs}
\label{sec:concentrate}
\new{With the tools we have built up, we can now start seriously analyzing the model.}
Assume from this point onwards that the initial condition is feasible.
Our goal in this section is to establish the following remarkable result \new{about concentration of beliefs}. In the long term, beliefs of agents in our model tend to concentrate at points of $\Theta$ which maximize the \emph{weighted likelihood function} $L$ \new{(Def. \ref{def:weightedlikelihood})}. Specifically, for a set $A \subseteq \Theta$ such that the values of $L$ on this set are dominated in some sense by $L(\theta')$ for another $\theta' \in \Theta$,
$$\int_{\theta\in A}f_i^{(t)}(\theta)\to 0$$
as $t\to\infty$. \new{We can think of this as saying that the agents' beliefs are being concentrated more and more near the maximizer of $L$, and so the beliefs on any set that stays clear of the maximizer tend to zero. The precise statement is in} the Master Theorem \ref{thm:master}.
\new{Once we have this result, it will become relatively easy to} show convergence of the agents' beliefs towards maximizers of $L$ \new{in various settings; this will be done} in Section \ref{sec:converge}.

\subsection{Weighted Likelihood}
\label{subsec:Lfun}
We first define the weighted likelihood function \new{mentioned in the introduction}.

\begin{definition}
\label{def:weightedlikelihood}
Assume that the initial condition is feasible. The \emph{weighted likelihood function} is defined for each $\theta \in \Theta$ by
$$L(\theta)= \prod_{i=1}^N (g_i^{(0)}(\theta))^{v_i}.$$
\end{definition}

Recall that \new{$g_i^{(0)}$ is the normalized beliefs defined in Definition \ref{def:normalized_belief}} and that the eigenvector centrality $v_i$ can be thought of as the influence of agent $i$ (Sec. \ref{subsec:perron}).
So this function aggregates the normalized beliefs of agents while being biased by the agents' influences.
\new{To see how this function is being ``biased,''
the \emph{unbiased} likelihood of each parameter $\theta$ given the initial beliefs of agents formed according to an underlying scenario (Sec. \ref{subsec:modelupdates}) is given by
$$L_{\text{unbiased}}(\theta) = c f_*(\theta) \prod_{i=1}^N g_i^{(0)}(\theta),$$
where $c$ is a normalizing constant. Under optimal information aggregation, the posterior belief is proportional to $L_{\text{unbiased}}$. To derive this, simply do the computation analogous to the one in Proposition \ref{prop:justifyrule}. 
Notice that $L$ and $L_{\text{unbiased}}$ are different in two ways. First, each $g_i^{(0)}$ factor in $L$ is to the power of $v_i$, suggesting that $L$ is biased by the centrality of each agent as already pointed out. Second, there is an extra factor of $f_*$, the prior, in $L_{\text{unbiased}}$. This is because agents in our model become so confident in their opinions that the prior is disregarded entirely, so this factor does not appear in $L$.}

\new{We may wonder what properties the weighted likelihood function $L$ has.}
The following curious proposition \new{gives a partial answer}: $L$ is almost bounded, that is, almost in $L^\infty$.
\new{This result is not used anywhere.}

\begin{prop}
\label{prop:LLp}
For any $1\leq p < \infty$, $L \in L^p_*$ with nonzero norm.
\end{prop}

\begin{shownto}{abb}
\begin{proofsketch}
By Proposition \ref{prop:update_t_g}, $\prod_{j=1}^N (g_j^{(0)})^{P_{ji}^{(t)}} \in L^1_*$.
By Lemma \ref{lem:newholder}, we can take the weighted geometric mean for all $i$ to get
$\prod_{j=1}^N (g_j^{(0)})^{\sum_{i=1}^N v_iP_{ji}^{(t)}} \in L^1_*.$
By Lemma \ref{lem:path_recurrence}, this means
$L \in L^{r^t}_*$ where $r=\rho(A)$.
Note that $r>1$ by Lemma \ref{lem:spectral_radius_greater_than_one} and interpolate using Lemma \ref{lem:interpolationLp}.
\end{proofsketch}
\end{shownto}

\begin{shownto}{full}
\begin{proof}
By Proposition \ref{prop:update_t_g}, $\prod_{j=1}^N (g_j^{(0)})^{P_{ji}^{(t)}} \in L^1_*$ for any $i$ and $t$.
\new{The idea is that we can average these functions together to get the desired result.}
Recall that $\sum_i v_i=1$. By Lemma \ref{lem:newholder}, we can take the weighted geometric mean for all $i$ weighted by $v_i$ to get
$$\prod_{j=1}^N (g_j^{(0)})^{\sum_{i=1}^N v_iP_{ji}^{(t)}} \in L^1_*.$$
By Lemma \ref{lem:path_recurrence}, $P_{ji}^{(t)}$ is the $(j,i)$ entry of $A^t$, so the quantity in the exponent is the $j$th entry of $A^tv$. Because $Av=rv$ where $r=\rho(A)$, this quantity is $r^t v_j$.
We conclude that $L \in L^{r^t}_*$ for all $t\geq 0$.
Because $r>1$ by Lemma \ref{lem:spectral_radius_greater_than_one}, Lemma \ref{lem:interpolationLp} implies that $L \in L^p_*$ for all $1 \leq p < \infty$. Its norms are nonzero because the initial condition is nondegenerate (Prop. \ref{prop:nondeg}).
\end{proof}
\end{shownto}

We now show that the weighted likelihood function guides the convergence of the ratio of beliefs at two different points (Prop. \ref{prop:maxdominate}).
\new{We can think of this as a ``pointwise'' concentration result. It is much weaker than the actual concentration result we seek (except when $\Theta$ is finite; see Theorem \ref{thm:mainfinite}).
Nonetheless, it is} the basis of the proof of the Master Theorem \ref{thm:master}.
The next lemma deals with the degeneracy of the weighted likelihood function, \new{i.e. the characterization of when $L(\theta)=0$. It is used to avoid some issues of dividing by zero.}

\begin{lemma}
\label{lem:Lzeroornot}
There is $T$ with the following property.
For any $\theta \in \Theta$ and $i$, if $L(\theta)=0$, then $f_i^{(t)}(\theta)=0$ for $t \geq T$, and if $L(\theta)>0$, then $f_i^{(t)}(\theta)> 0$ for all $t \geq 0$.
\end{lemma}

\begin{shownto}{abb}
\begin{proofsketch}
Straightforward by Proposition \ref{prop:update_t_g}.
\end{proofsketch}
\end{shownto}

\begin{shownto}{full}
\begin{proof}
If $L(\theta)>0$, then $f_i^{(0)}(\theta)>0$ for all $i$. The desired result follows from \new{the formula in} Proposition \ref{prop:update_t_g}.
\new{Now suppose that $L(\theta)=0$.}
Let $T$ be large enough so that $P_{ij}^{(t)}>0$ for all $i,j$ and $t\geq T$.
\new{Because} $f_j^{(0)}(\theta)=0$ for some $j$, $f_i^{(t)}(\theta)=0$ for all $i$ and $t\geq T$ by \new{the formula in} Proposition \ref{prop:update_t_g}.
\end{proof}
\end{shownto}

\new{The next proposition gives an explicit formula for the ratio of beliefs at two different points, highlighting the terms that become small in the limit. It is a crucial component that is used in Section \ref{subsec:master} to prove the Master Theorem \ref{thm:master}.}

\begin{prop}
\label{prop:ratio_of_f_asymptotic}
Let $r=\rho(A)$ be the spectral radius of $A$, and $v$ and $w$ be the Perron vectors of $A$ and $A^\top$, respectively.
Then there are $\eps_{i,j}^{(t)} \to 0$ as $t\to \infty$ such that for any $\theta_1,\theta_2 \in \Theta$ with $L(\theta_2)>0$,
$$\frac{f_i^{(t)}(\theta_1)}{f_i^{(t)}(\theta_2)}
= \frac{f_*(\theta_1)}{f_*(\theta_2)}
\bracket{
\prod_{j=1}^{N}
\Bigg( \frac{g_j^{(0)}(\theta_1)}{g_j^{(0)}(\theta_2)}\Bigg)^{v_j+\eps_{i,j}^{(t)}}}
^{w_i r^t/w^\top v}.$$
\end{prop}

\new{From this proposition, we can begin to see the weighted likelihood function take shape: in the formula above, if the prior factors in front are ignored and the small terms $\eps_{i,j}^{(t)}$ disregarded, what is left is the ratio $L(\theta_1)/L(\theta_2)$ raised to some power.}

\begin{shownto}{abb}
\begin{proofsketch}
Use Proposition \ref{prop:update_t_g} to get the expression for the ratio
$g_i^{(t)}(\theta_1)/g_i^{(t)}(\theta_2)$.
Then use Lemma \ref{lem:path_recurrence} and Theorem \ref{thm:perron} to transform it into the desired formula.
\end{proofsketch}
\end{shownto}

\begin{shownto}{full}
\begin{proof}
\new{The proof is essentially by direct computation.}
Note that $f_i^{(t)}(\theta_2)>0$ by Lemma \ref{lem:Lzeroornot}, \new{so the formula makes sense}. By Proposition \ref{prop:update_t_g} for $\theta_1$ and $\theta_2$,
$$ \frac{g_i^{(t)}(\theta_1)}{g_i^{(t)}(\theta_2)} = \prod_{j=1}^{N} \Bigg( \frac{g_j^{(0)}(\theta_1)}{g_j^{(0)}(\theta_2)} \Bigg)^{P_{ji}^{(t)}}.$$
By Lemma \ref{lem:path_recurrence}, $P_{ji}^{(t)}$ is the $(j,i)$ entry of $A^t$.
By Theorem \ref{thm:perron}, this quantity equals
$$r^t \left( \frac{(vw^\top)_{ji}}{w^\top v} + \delta_{i,j}^{(t)} \right) = r^t \left( \frac{v_jw_i}{w^\top v} + \delta_{i,j}^{(t)} \right),$$ where $\delta_{i,j}^{(t)} \to 0$ as $t\to \infty$. The result follows by letting $\eps_{i,j}^{(t)} = \delta_{i,j}^{(t)} w^\top v/w_i$ \new{and substituting these all in}.
\end{proof}
\end{shownto}

\new{We can now show the ``pointwise'' concentration of beliefs at maximizers of $L$. This hints at the actual concentration result in the next subsection.}

\begin{prop}
\label{prop:maxdominate}
If $\theta_1,\theta_2 \in \Theta$ are such that $L(\theta_1) < L(\theta_2)$, then as $t \to \infty$,
$$\frac{f_i^{(t)}(\theta_1)}{f_i^{(t)}(\theta_2)} \to 0.$$
\end{prop}

\begin{shownto}{abb}
\begin{proofsketch}
Use Proposition \ref{prop:ratio_of_f_asymptotic} and the fact that $r>1$ from Lemma \ref{lem:spectral_radius_greater_than_one}.
\end{proofsketch}
\end{shownto}

\begin{shownto}{full}
\begin{proof}
\new{The proof is by directly applying} Proposition \ref{prop:ratio_of_f_asymptotic}.
As $t\to\infty$, the expression in the bracket converges to $L(\theta_1)/L(\theta_2)<1$, so this expression is smaller than some $\alpha<1$ for large $t$. By Lemma \ref{lem:spectral_radius_greater_than_one}, the exponent of this expression goes to infinity, so $f_i^{(t)}(\theta_1)/f_i^{(t)}(\theta_2)\to 0$ as $t\to\infty$.
\end{proof}
\end{shownto}

\subsection{Master Theorem}
\label{subsec:master}

\new{This is a technical section; its aim is to prove the Master Theorem \ref{thm:master} which establishes concentration of beliefs in very general settings. In order to do this, subtle issues of convergence arise which can be best dealt with using ideas from measure theory (e.g. Lebesgue's dominated convergence theorem). Even in the relatively benign setting $\Theta = \mb{N}$, these issues are not so easy to circumnavigate.}
We first present the \new{general} idea of the Master Theorem. Suppose $A\subseteq \Theta$ and $\theta' \in \Theta$ are such that $L(\theta)<L(\theta')$ for all $\theta \in A$.
\new{Thus, $\theta'$ ``dominates'' all of $A$ in the sense of the weighted likelihood function. We hope that
$\int_{\theta\in A} f_i^{(t)}(\theta) \to 0$ as $t\to\infty$.} Proposition \ref{prop:maxdominate} implies the \new{pointwise result} $f_i^{(t)}(\theta)/f_i^{(t)}(\theta') \to 0$ as $t\to\infty$. We would like to aggregate this pointwise result into
\begin{equation}
\label{eq:firststepmaster}
\frac{\int_{\theta\in A} f_i^{(t)}(\theta)}{f_i^{(t)}(\theta')} \to 0
\end{equation}
(Prop. \ref{prop:lessthanapoint}), so that we may conclude $\int_{\theta\in A}f_i^{(t)}(\theta) \to 0$.
To do so, we need to control the convergence in Proposition \ref{prop:maxdominate}; \new{the convergence needs to be ``uniform'' in some way}. It turns out that the correct hypothesis is that the inequality $L(\theta)<L(\theta')$ holds when the exponents in $L$ are slightly perturbed: $L(\theta)g_i^{(0)}(\theta)^{\delta_i}<L(\theta')g_i^{(0)}(\theta')^{\delta_i}$ for small $\delta_i>0$.
\new{As explained in Section \ref{subsec:contribute}, we can think of this as $\theta'$ dominating $A$ \emph{robustly}, i.e. the domination does not fail even if the centralities of agents are being slightly perturbed.
Later we will offer an example where this robustness condition fails and concentration also fails} (Ex. \ref{ex:countercountable}), which shows that this hypothesis cannot be omitted.

After obtaining the bound \eqref{eq:firststepmaster}, if $\Theta$ is discrete, \new{i.e. $\Theta = \mb{N}$}, then the obvious bound $f_i^{(t)}(\theta')\leq 1$ finishes the proof. \new{But this is not so easy if $\Theta$ is continuous, i.e. $\Theta = \R^k$. In fact, if concentration is to be true, it is likely that $f_i^{(t)}(\theta')\to \infty$ as $t\to\infty$ if $\theta'$ maximizes $L$, and so we cannot get from \eqref{eq:firststepmaster} to the desired result directly. The second idea is to pick $\theta'$ that is \emph{not} a maximizer of $L$ and show that $f_i^{(t)}(\theta') \to 0$ (which is likely to happen if concentration near the maximizer is to occur), which will then allow us to get the desired result.
To do this, we} find a set $B$ such that $L(\theta)>L(\theta')$ for all $\theta \in B$, \new{i.e. now the set $B$ dominates $\theta'$ instead}. Then if the convergence is similarly controlled, we should get
$$\frac{\int_{\theta\in B}f_i^{(t)}(\theta)}{f_i^{(t)}(\theta')} \to \infty$$
(Prop. \ref{prop:morethanapoint}), which implies $f_i^{(t)}(\theta') \to 0$. Putting all of this together completes the proof.

\new{To summarize, the mechanics of the Master Theorem are a two-step domination process. Suppose that we have sets $A$ and $B$ and a point $\theta'$, such that $\theta'$ dominates $A$ and $B$ dominates $\theta'$, and both dominations are robust. Then the latter domination forces the density at $\theta'$ to go to zero, which, by the former domination, in turn forces the density on $A$ to go to zero. To see how these sets might be picked in actual situations, suppose $\Theta=\R$ and $L$ is continuous with a unique peak at 0. Let $A=[1,2]$. Then it might be possible to pick $\theta'=1/2$ and $B=[-1/4,1/4]$, i.e. we want to pick $B$ to be very close to 0, and $\theta'$ to be somewhere between $A$ and $B$. This selection process will be further automated in Section \ref{sec:converge} to get results like Theorem \ref{thm:mainsimplified} that can be directly used in practice.}

\new{In what follows, we implement the program outlined above. First we formalize the notion of this ``robust domination'' which allows us to control convergence. The following lemma is a purely analytical statement that contains the essential ingredients to carrying out the ``$\theta'$ dominates $A$'' part; then Proposition \ref{prop:lessthanapoint} will actually carry it out. Throughout,} we assume familiarity with Lebesgue's dominated convergence theorem and Fatou's lemma.

\begin{lemma}
\label{lem:dominatedbound}
Let $S$ be a measure space. For $t \in \Z_{\geq 0}$ and $1\leq i \leq N$, let $a_i,c:S \to \R_{\geq 0}$ be measurable functions, $r_t >0$ and $\eps_i^{(t)} \geq -1$.
As $t\to\infty$, $r_t \to \infty$ and, for every $i$, $\eps_i^{(t)}\to 0$.
The functions $a_i$ has product
$A(x):=\prod_{i=1}^N a_i(x)\leq 1$ for almost every $x \in S$.
Moreover, for every $i$, there is $\delta_i > 0$ such that $A(x) a_i(x)^{\delta_i} \leq 1$ for almost every $x \in S$.
Suppose that for every $t$, the integral
$$I_t:= \int_{x \in S} c(x) \paren{\prod_{i=1}^N  a_i(x)^{1+\eps_i^{(t)}}}^{r_t} <\infty.$$
Then, as $t\to\infty$,
$$I_t \to \int_{A(x)=1} c(x)<\infty.$$
In particular, if $A(x)<1$ almost everywhere, then $I_t \to 0$ as $t\to\infty$.
\end{lemma}

\new{Note that the form of $I_t$ is reminiscent of the expression in Proposition \ref{prop:ratio_of_f_asymptotic} and indeed this is what the lemma will be used for; but this abstraction in terms of these various variables is necessary in order for the proof of the lemma to not become too large and convoluted.}

\begin{shownto}{abb}
\begin{proofsketch}
The $A(x) a_i(x)^{\delta_i} \leq 1$ hypothesis can be strengthened to the following. There is $0<\delta'<1$ such that for all $\abs{y_i}<\delta'$,
for almost every $x \in S$,
$$A(x)^{(3+\delta')/2}\leq \prod_{i=1}^N a_i(x)^{1+y_i} \leq A(x)^{(1-\delta')/2}.$$
So there is $T$ such that for $t\geq T$,
$$\int_{x\in S} c(x) A(x)^{r_t(3+\delta')/2} \leq I_t \leq \int_{x\in S} c(x) A(x)^{r_t(1-\delta')/2}\leq \infty.$$
Use the fact that $I_t<\infty$ and $r_t(1-\delta')>r_T(3+\delta')$ for large $t$ to show that both integrals converge to the desired limit by Lebesgue's dominated convergence theorem.
\end{proofsketch}
\end{shownto}

\begin{shownto}{full}
\begin{proof}
\new{We have $\prod_{i=1}^N a_i(x) = A(x)$. The first part of the proof is to show that if we perturb the exponents of $a_i$'s a little, the result will still be close to $A(x)$; specifically, it will be $A(x)$ raised to some power in a fixed range. To do this, the condition $A(x) a_i(x)^{\delta_i} \leq 1$ will be of help, along with a series of algebraic manipulations.}

All statements in this paragraph are meant to hold for almost every $x \in S$. For any $0\leq y<\delta_i/2$,
$$A(x) a_i(x)^y=A(x)^{1-y/\delta_i}(A(x) a_i(x)^{\delta_i})^{y/\delta_i}\leq A(x)^{1/2},$$
because $1-y/\delta_i>1/2$.
So if $\delta=\min \delta_i/2$, then
for any $0\leq y<\delta$, $A(x) a_i(x)^y\leq A(x)^{1/2}$ for all $i$.
\new{This allows us to perturb one exponent.
To perturb many exponents at once, note that} for any $0 \leq y_i<\delta/N$,
$$\prod_{i=1}^N a_i(x)^{1+y_i}
=\prod_{i=1}^N (A(x) a_i(x)^{Ny_i})^{1/N}\leq A(x)^{1/2}.$$
\new{Now we extend to the case where the perturbations can be negative.}
Let $0<\delta'<1$ be such that
$(1+\delta')/(1-\delta')=1+\delta/N$. Then for any $\abs{y_i}< \delta'$, if $y_j = \min y_i$,
$$\prod_{i=1}^N a_i(x)^{1+y_i}
=\bigg( \prod_{i=1}^N a_i(x)^{(1+y_i)/(1+y_j)}\bigg)^{1+y_j}\leq A(x)^{(1+y_j)/2}\leq A(x)^{(1-\delta')/2},$$
because $1 \leq (1+y_i)/(1+y_j) < 1+\delta/N$.
\new{This gives us an upper bound. Finally, to get a lower bound as well,} because
$\prod_{i=1}^N a_i(x)^{1+y_i} \prod_{i=1}^N a_i(x)^{1-y_i}=A(x)^2$,
for any $\abs{y_i}< \delta'$,
$$A(x)^{(3+\delta')/2}\leq \prod_{i=1}^N a_i(x)^{1+y_i} \leq A(x)^{(1-\delta')/2}.$$
\new{This final inequality is what we want: both upper and lower bounds for when all exponents are perturbed at once, regardless of signs.}

\new{The second part of the proof is to use this bound to establish convergence.}
Let $T$ be such that for every $i$ and $t\geq T$, $\big\lvert\eps_i^{(t)}\big\rvert<\delta'$. Then, for $t\geq T$, \new{by the previous paragraph},
$$\int_{x\in S} c(x) A(x)^{r_t(3+\delta')/2} \leq I_t \leq \int_{x\in S} c(x) A(x)^{r_t(1-\delta')/2}.$$
Note that, \new{as written}, the right-hand expression may be infinite.
\new{The final step is Lebesgue's dominated convergence.}
Let
$$f_t(x)=c(x) A(x)^{r_t(3+\delta')/2},\quad
g_t(x)=c(x) A(x)^{r_t(1-\delta')/2}.$$
Then $f_T \in L^1$ and  $f_t, g_t\to c 1_{A(x)=1}$ pointwise a.e. as $t\to\infty$. For large $t$, $r_t(1-\delta')>r_T(3+\delta')$, so $f_t \leq g_t \leq f_T$ pointwise a.e.
\new{Thus $f_T$ serves as the ``Lebesgue dominant'' in this argument.}
By Lebesgue's dominated convergence theorem, $c1_{A(x)=1} \in L^1$ and $\int_{x} f_t, \int_{x} g_t\to \int_{A(x)=1} c(x)$ as $t\to\infty$,
so $I_t$ converges to this same limit.
\end{proof}
\end{shownto}

\new{To better understand how the proof of Lemma \ref{lem:dominatedbound} works,
one might want to carry it out in the case where $S=\mb{N}$ with counting measure; in that case, the lemma is a statement about convergence of infinite series. The proof in this special case does not require Lebesgue's dominated convergence, but still requires some thought on how to estimate the tails of the series. The authors of this paper proved the lemma for this special case first before generalizing it to the present form; the final step is the only part of the argument that needs nontrivial change.}

\new{Once we have Lemma \ref{lem:dominatedbound}, Proposition \ref{prop:lessthanapoint} follows immediately, although it may be a little tedious to check that all conditions are satisfied.}

\begin{prop}
\label{prop:lessthanapoint}
Let $A\subseteq \Theta$ be a measurable set. Let $\theta' \in \Theta$ be such that $L(\theta')>0$ and $L(\theta) \leq L(\theta')$ for almost every $\theta \in A$. Moreover, for each $i$, there is $\delta_i>0$ such that $L(\theta)g_i^{(0)}(\theta)^{\delta_i} \leq L(\theta')g_i^{(0)}(\theta')^{\delta_i}$ for almost every $\theta \in A$.
Then there is $0\leq M<\infty$ such that, as $t\to\infty$,
$$\frac{\int_{\theta\in A}f_i^{(t)}(\theta)}{f_i^{(t)}(\theta')} \to M.$$
If the set $\{\theta \in A: L(\theta) = L(\theta') \}$ has measure zero, then $M=0$.
\end{prop}

\begin{shownto}{abb}
\begin{proofsketch}
Apply Lemma \ref{lem:dominatedbound} to the asymptotic formula in Proposition \ref{prop:ratio_of_f_asymptotic}.
\end{proofsketch}
\end{shownto}

\begin{shownto}{full}
\begin{proof}
\new{All we have to do is to directly apply Lemma \ref{lem:dominatedbound}.}
Apply Proposition \ref{prop:ratio_of_f_asymptotic} to get the expression into the form
$$\int_{\theta \in A} \frac{f_*(\theta)}{f_*(\theta')}
\bracket{
\prod_{j=1}^{N}
\Bigg( \frac{g_j^{(0)}(\theta)}{g_j^{(0)}(\theta')}\Bigg)^{v_j+\eps_{i,j}^{(t)}}}
^{w_i r^t/w^\top v}.$$
Then apply Lemma \ref{lem:dominatedbound} with $S=A$, $$c(\theta) = \frac{f_*(\theta)}{f_*(\theta')},\quad a_j(\theta) = \Bigg( \frac{g_j^{(0)}(\theta)}{g_j^{(0)}(\theta')}\Bigg)^{v_j}, \quad \eps_j^{(t)}=\frac{\eps_{i,j}^{(t)}}{v_j}, \quad r_t = \frac{w_i r^t}{w^\top v}.$$
The hypotheses of the lemma are readily verified using the hypotheses of the proposition and Lemma \ref{lem:spectral_radius_greater_than_one}.
\end{proof}
\end{shownto}

\new{The previous proposition completes the ``$\theta'$ dominates $A$'' part of the program. To carry out the ``$B$ dominates $\theta'$'' part, we prove the following Lemma  \ref{lem:reversebound} and Proposition \ref{prop:morethanapoint}, which can be thought of as analogous versions of Lemma \ref{lem:dominatedbound} and Proposition \ref{prop:lessthanapoint} but 
 with reversed inequalities and slight changes in details.
 In particular, the part that changes is the final step of the lemma where here we use Fatou's lemma to conclude instead.}

\begin{lemma}
\label{lem:reversebound}
Let $S,a_i,c,r_t,\eps_i^{(t)}, A$ and $I_t$ be as in Lemma \ref{lem:dominatedbound} with the following changes.
The inequalities involving $A$ are reversed:
$A(x) \geq 1$ for almost every $x\in S$, and for every $i$, there is $\delta_i > 0$ such that $A(x) a_i(x)^{\delta_i} \geq 1$ for almost every $x \in S$.
Moreover, the integral $I_t$ may be infinite.
If $A(x)=1$ almost everywhere, then as $t\to\infty$,
$$I_t \to \int_{x \in S} c(x),$$
where the limit can be infinite. Otherwise, $I_t \to \infty$ as $t\to\infty$.
\end{lemma}

\begin{shownto}{abb}
\begin{proofsketch}
Similarly to Lemma \ref{lem:dominatedbound}, there are $0<\delta'<1$ and $T$ such that, for $t\geq T$,
$$\int_{x\in S} c(x) A(x)^{r_t(3+\delta')/2} \geq I_t \geq \int_{x\in S} c(x) A(x)^{r_t(1-\delta')/2}.$$
Apply Fatou's lemma to the right-hand integrand.
\end{proofsketch}
\end{shownto}

\begin{shownto}{full}
\begin{proof}
\new{The first part of the proof stays the same.}
By an analogous argument to the one in Lemma \ref{lem:dominatedbound} with reversed inequalities, there are $0<\delta'<1$ and $T$ such that, for $t\geq T$,
$$\int_{x\in S} c(x) A(x)^{r_t(3+\delta')/2} \geq I_t \geq \int_{x\in S} c(x) A(x)^{r_t(1-\delta')/2}.$$
\new{To carry out the second part of the proof, use the following argument.}
The result is immediate if $A(x)=1$ almost everywhere. Otherwise, the right-hand integrand converges pointwise to a function that is infinite on a set of positive measure. By Fatou's lemma, the right-hand integral converges to infinity. Then $I_t \to \infty$ as $t \to\infty$.
\end{proof}
\end{shownto}

One can prove the following proposition by applying Lemma \ref{lem:reversebound} in the same way that one applies Lemma \ref{lem:dominatedbound} to prove Proposition \ref{prop:lessthanapoint}.
\new{So we omit the proof of the following proposition.}

\begin{prop}
\label{prop:morethanapoint}
Let $B\subseteq \Theta$ be a measurable set. Let $\theta' \in \Theta$ be such that $L(\theta')>0$ and $L(\theta) \geq L(\theta')$ for almost every $\theta \in B$. Moreover, for each $i$, there is $\delta_i>0$ such that $L(\theta)g_i^{(0)}(\theta)^{\delta_i} \geq L(\theta')g_i^{(0)}(\theta')^{\delta_i}$ for almost every $\theta \in B$.
Then there is $0\leq M\leq \infty$ such that, as $t\to\infty$,
$$\frac{\int_{\theta\in B}f_i^{(t)}(\theta)}{f_i^{(t)}(\theta')} \to M.$$
If $B$ has positive measure, then $M>0$. If $\{\theta \in B: L(\theta) > L(\theta') \}$ has positive measure, then $M=\infty$.
\end{prop}

\new{From here it is easy to combine Propositions \ref{prop:lessthanapoint} and \ref{prop:morethanapoint} into the Master Theorem. Recall that the various technical hypotheses of the theorem should be intuitively read as ``$\theta'$ dominates $A$, and $B$ dominates $\theta'$ robustly in the sense of the weighted likelihood function,'' as explained earlier in this section.}

\begin{theorem}[Master Theorem]
\label{thm:master}
Let $A,B\subseteq \Theta$ be measurable sets, where $B$ has positive measure. Let $\theta' \in \Theta$ be such that $L(\theta')>0$, $L(\theta) \leq L(\theta')$ for almost every $\theta \in A$, and $L(\theta) \geq L(\theta')$ for almost every $\theta \in B$. Moreover, for each $i$, there are $\delta_i,\delta'_i>0$ such that $L(\theta)g_i^{(0)}(\theta)^{\delta_i} \leq L(\theta')g_i^{(0)}(\theta')^{\delta_i}$ for almost every $\theta \in A$ and $L(\theta)g_i^{(0)}(\theta)^{\delta'_i} \geq L(\theta')g_i^{(0)}(\theta')^{\delta'_i}$ for almost every $\theta \in B$.
Suppose that either $\{\theta \in A: L(\theta) = L(\theta') \}$ has measure zero or $\{\theta \in B: L(\theta) > L(\theta') \}$ has positive measure.
Then
$$\int_{\theta\in A}f_i^{(t)}(\theta) \to 0.$$
\end{theorem}

\begin{shownto}{abb}
\begin{proofsketch}
Propositions \ref{prop:lessthanapoint} and \ref{prop:morethanapoint} imply that
$\int_{\theta\in A}f_i^{(t)}(\theta)\big/\int_{\theta\in B}f_i^{(t)}(\theta) \to 0.$
\end{proofsketch}
\end{shownto}

\begin{shownto}{full}
\begin{proof}
By Propositions \ref{prop:lessthanapoint} and \ref{prop:morethanapoint},
$$\frac{\int_{\theta\in A}f_i^{(t)}(\theta)}{\int_{\theta\in B}f_i^{(t)}(\theta)} \to \frac{M_1}{M_2},$$
where $M_1$ and $M_2$ are limits in the two respective propositions. Because $B$ has positive measure, $M_2>0$. Our final hypothesis implies that either $M_1=0$ or $M_2=\infty$, so $M_1/M_2=0$. Then our result follows from $\int_{\theta\in B}f_i^{(t)}(\theta) \leq 1$.
\end{proof}
\end{shownto}

\section{Convergence of Beliefs}
\label{sec:converge}

We apply the Master Theorem \ref{thm:master} to show convergence of beliefs of agents toward a point distribution at the maximizer or near-maximizer of $L$.
We establish separate results for the finite case (Sec. \ref{subsec:finite}), the infinite discrete case (Sec. \ref{subsec:infinite}) and the $\R^k$ case (Sec. \ref{subsec:Rk}). Our results in the $\R^k$ case are the basis of Theorem \ref{thm:mainsimplified}. Results in this section will be used in applications in Section \ref{sec:apps}.

\new{We first explain the general idea of this section. Recall that to use the Master Theorem, for a given set $A$, we need to find a point $\theta'$ that dominates $A$ and a set $B$ that dominates $\theta'$.
What this section aims to do is the following.
First, for a given set $A$, the choice of $\theta'$ and $B$ will be automated and abstracted away from the user.
Second, even the choice of $A$ will be abstracted away, leaving only the conclusion that ``beliefs converge to a point distribution,'' at least in the case of a unique maximizer. The general idea is not hard, but the details can be quite formidable.}

\subsection{Preliminaries}

\new{First we need some preliminaries. This subsection presents the technical details in carrying out the ``automatically choosing $\theta'$ and $B$'' part.} Let $L_{\esssup}=\esssup_{\theta \in \Theta} L(\theta) \leq \infty$, that is, $L_{\esssup}$ is the smallest $M$ for which $\{\theta \in \Theta: L(\theta)>M\}$ has measure zero.
\new{Note that} we can replace ``$\esssup$'' by ``$\sup$'' if $\Theta$ is infinite discrete or ``$\max$'' if $\Theta$ is finite.
The following proposition gives sufficient conditions for beliefs to concentrate in places where $L$ is near $L_{\esssup}$.

\begin{prop}
\label{prop:nearessup}
Let $0\leq M < L_{\esssup}$ and let $S=\{\theta \in \Theta: L(\theta)>M\}$.
Assume that for every $i$, $g_i^{(0)}$ is bounded a.e. on $\Theta \setminus S$.
Suppose that one of the following holds:
\begin{enumerate}[wide]
    \item $L$ is not constant a.e. on $S$ and for every $i$, $g_i^{(0)}$  is bounded below by a positive number a.e. on $S$;
    \item there is a subset $T\subseteq S$ of positive measure such that, for every $i$, $g_i^{(0)}$ is constant on $T$.
\end{enumerate}
Then, as $t\to\infty$,
$$\int_{\theta \in S} f_i^{(t)}(\theta) \to 1.$$
\end{prop}

\new{To better understand the above proposition, it will be useful to understand what roles the various hypotheses play. The ``bounded'' and ``bounded below by a positive number'' hypotheses ensure that the gaps between $A$ and $\theta'$ and $\theta'$ and $B$ are robust as needed in the Master Theorem \ref{thm:master}. Now we note that $S$ dominates $\Theta \setminus S$; however, we still need to pick $\theta'$ and $B$ out of $S$. There are two ways to do this depending on the situation. First, think of $\Theta = \R$. In this case, if $L$ is constant on $S$, it will be hard to pick $B$ that dominates $\theta'$ robustly; but we can do this otherwise. This is condition (1) in the proposition. Second, think of $\Theta = \mb{N}$.
In this case, even if $S$ is a single point (and so $L$ is constant on it), it will be easy to pick $\theta'$ and $B$; but this case does not satisfy (1). This is where condition (2) comes in.
Putting all these details together gives us the proposition.}

\begin{shownto}{abb}
\begin{proofsketch}
Apply the Master Theorem \ref{thm:master}. Let $A=\Theta \setminus S$.
In the second case, pick any $\theta' \in T$ and let $B=T$.
Now consider the first case. If there is $\theta_1 \in S$ that minimizes $L$ in $S$,
there is $M'>L(\theta_1)$ such that $U=\{\theta \in \Theta: L(\theta)> M'\}$ has positive measure. Let $\theta'=\theta_1$ and $B=U$.
Otherwise, there are $\theta_1,\theta_2 \in S$ such that $S_1=\{\theta \in \Theta: L(\theta)\geq L(\theta_1)\}$ has positive measure and $L(\theta_2)<L(\theta_1)$. Let $\theta'=\theta_2$  and $B=S_1$.
\end{proofsketch}
\end{shownto}

\begin{shownto}{full}
\begin{proof}
Note that $S$ has positive measure \new{by the definition of $L_{\esssup}$}. Apply the Master Theorem \ref{thm:master} \new{with} $A=\Theta \setminus S$.
In the second case, \new{we can} pick any $\theta' \in T$ and let $B=T$.
Now consider the first case. \new{First suppose that} there is $\theta_1 \in S$ that minimizes $L$ in $S$. In this case we pick $\theta'=\theta_1$; \new{now we need to pick $B$}.
Because $L$ is not constant a.e. on $S$, $\{\theta \in \Theta: L(\theta)> L(\theta_1)\}$ has positive measure.
By continuity of measure from below, there is $M'>L(\theta_1)$ such that $U=\{\theta \in \Theta: L(\theta)> M'\}$ has positive measure. Let $B=U$.
\new{Finally, suppose that there is no minimizer of $L$ on $S$.} Pick a sequence $\theta_i \in S$ such that $L(\theta_i) \to \inf_{\theta\in S} L(\theta)$. Define $S_i=\{\theta \in \Theta: L(\theta)\geq L(\theta_i)\}$.
Then $S=\cup_i S_i$.
By subadditivity of measure, one of $S_i$ has positive measure. \new{Then for this $i$,} pick $j$ such that $L(\theta_j)<L(\theta_i)$. Let $\theta'=\theta_j$  and $B=S_i$.

The result \new{now} follows by applying the Master Theorem to $A$, $B$ and $\theta'$ \new{just chosen}.
\new{To show} the existence of $\delta_i>0$, use the fact that $L(\theta)/L(\theta')\leq \alpha<1$ for all $\theta \in A$ and $g_i^{(0)}$ is bounded a.e. on $A$. The existence of $\delta'_i>0$ is similar; \new{in case (1) we use the fact that $L(\theta)/L(\theta')\geq \beta > 1$ for all $\theta \in B$ and $g_i^{(0)}$ is bounded below a positive number a.e. on $S$.}
\end{proof}
\end{shownto}

\new{We now present the following example of a social network}, which will be used in Examples \ref{ex:finite} and \ref{ex:countercountable}.
\new{What this example is trying achieve is a network consisting of \emph{two} people, but where the strength of information flow on each edge can vary, i.e. a ``weighted'' social network.
We can generalize our model to allow for weighted edges, but we choose not to: it turns out that, at least in this simple case, the weighted edges can be simulated by our unweighted network by using teams of people to represent one person in the weighted case. The details are below.}

\begin{example}
\label{ex:abgraph}
For $a,b\geq 1$, \new{we define} the \emph{$(a,b)$-graph} as follows. Let $c=\max(a,b)$. The graph has $2c$ nodes, $x_1,\dots,x_c$ and $y_1,\dots,y_c$. There is an edge from $x_i$ to $x_j$ and from $y_i$ to $y_j$ if and only if $j-i \equiv 0,1,\dots,a-1 \pmod{c}$; there is an edge from $x_i$ to $y_j$ and from $y_i$ to $x_j$ if and only if $j-i \equiv 0,1,\dots,b-1 \pmod{c}$. \new{This graph is supposed to represent two people, $x$ and $y$, with self-loops at $x$ and $y$ each of weight $a$ and an undirected edge from $x$ to $y$ of weight $b$.} Let the $x_i$'s and $y_i$'s have equal initial beliefs $f_x^{(0)}$ and $f_y^{(0)}$, respectively. Then \new{it is easy to see that} the $x_i$'s and $y_i$'s have equal beliefs throughout, so we can \new{just} subscript with $x$ and $y$.
\begin{shownto}{full}
The weighted likelihood function is \new{then}
$L(\theta)=g_x^{(0)}(\theta)^{1/2}g_y^{(0)}(\theta)^{1/2}$.
By Proposition \ref{prop:update_t_g}, \new{we can compute}
$$f_x^{(t)}
=\frac{f_{\theta^*} \big(g_x^{(0)}g_y^{(0)}\big)^{(a+b)^t/2} \big( g_x^{(0)}/g_y^{(0)}\big)^{(a-b)^t/2}
}{\int_{\Theta} f_{\theta^*} \big(g_x^{(0)}g_y^{(0)}\big)^{(a+b)^t/2} \big(g_x^{(0)}/g_y^{(0)}\big)^{(a-b)^t/2}
}.$$
\end{shownto}
\end{example} 

\subsection{Finite Case}
\label{subsec:finite}

\new{We now apply Proposition \ref{prop:nearessup}
to show convergence in the case where $\Theta$ is finite.}
The theorem \new{in this case is quite easy and we can appeal to either case (2) of} Proposition \ref{prop:nearessup} or \new{just} Proposition \ref{prop:maxdominate}.

\begin{theorem}
\label{thm:mainfinite}
Assume that $\Theta$ is finite and the initial condition is feasible. Let $\theta_{\max}$ be the set of points that maximize $L$. Then as $t\to\infty$,
$$\sum_{\theta \in \theta_{\max}} f_i^{(t)}(\theta) \to 1.$$
In particular, if $\theta_{\max}$ consists of a single point, then $f_i^{(t)}$ converges to the point distribution at that point.
\end{theorem}

One may wonder what more can be said when $\theta_{\max}$ consists of several points.
The following example shows that the behavior can greatly vary: \new{with two maximizers, the maximizers may share a constant fraction of beliefs, the beliefs may converge towards one of the maximizers, the fraction of beliefs between the maximizers may alternate between two different values, and the beliefs may converge separately on odd and even time steps to both maximizers!}

\begin{example}
\label{ex:finite}
Consider the $(a,b)$-graph of Example \ref{ex:abgraph}. Let \new{the parameter space} $\Theta=\{0,1\}$ with \new{prior} $f_*\equiv 1$ \new{and initial beliefs} $f_x^{(0)}(0)=\alpha$ and $f_y^{(0)}(0)=1-\alpha$,
where $0<\alpha<1$.
\begin{shownto}{abb}
Then $\theta_{\max}=\Theta$. We omit verification of the following.
For $(a,b)=(2,1)$, $f_x^{(t)}(0)$ is constant at $\alpha$. For $(a,b)=(1,2)$, $f_x^{(t)}(0)$ alternates between $\alpha$ and $1-\alpha$.
For $(a,b)=(3,1)$, $f_x^{(t)}(0)\to 0$ if $\alpha<1/2$ and $f_x^{(t)}(0)\to 1$ if 
$\alpha>1/2$.
For $(a,b)=(1,3)$ and $\alpha<1/2$, $f_x^{(2t)}(0)\to 0$ while $f_x^{(2t+1)}(0)\to 1$.
\end{shownto}
\begin{shownto}{full}
Then $\theta_{\max}=\Theta$, and \new{by the computation in Example \ref{ex:abgraph},}
$$f_x^{(t)}(0)=\frac{1}{1+((1-\alpha)/\alpha)^{(a-b)^t}}.$$
\new{We can thus see the following behaviors.} For $(a,b)=(2,1)$, $f_x^{(t)}(0)$ is constant at $\alpha$. For $(a,b)=(1,2)$, $f_x^{(t)}(0)$ alternates between $\alpha$ and $1-\alpha$.
For $(a,b)=(3,1)$, $f_x^{(t)}(0)\to 0$ if $\alpha<1/2$ and $f_x^{(t)}(0)\to 1$ if 
$\alpha>1/2$.
For $(a,b)=(1,3)$ and $\alpha<1/2$, $f_x^{(2t)}(0)\to 0$ while $f_x^{(2t+1)}(0)\to 1$.
\end{shownto}
\end{example}

\subsection{Infinite Discrete Case}
\label{subsec:infinite}

\new{We now move on to} the theorem for the infinite discrete case, \new{which} follows from the second case of Proposition \ref{prop:nearessup}. This is the first nontrivial case \new{and contains hypotheses that the the initial conditions have to verify in order for beliefs to converge.}

\begin{theorem}
\label{thm:maincountable}
Assume that $\Theta$ is infinite discrete and the initial condition is feasible.

\begin{enumerate}[wide]
    \item Let $L$ attain the maximum $L_{\max}$ at the set of points $\theta_{\max}\neq\varnothing$.
Suppose that $\sup_{\theta \notin \theta_{\max}} L(\theta) < L_{\max}$ and $g_i^{(0)}$ is bounded on $\Theta\setminus \theta_{\max}$.
Then as $t\to\infty$,
$$\sum_{\theta \in \theta_{\max}} f_i^{(t)}(\theta) \to 1.$$
In particular, if $\theta_{\max}$ consists of a single point, then $f_i^{(t)}$ converges to the point distribution at that point.
    \item If $L$ does not attain a maximum, let $L_{\sup}=\sup_{\theta\in\Theta} L(\theta)$.
For any $M<L_{\sup}$, if $g_i^{(0)}$is bounded on $\{\theta \in\Theta: L(\theta)\leq M\}$,
then as $t\to\infty$,
$$\sum_{L(\theta)>M} f_i^{(t)}(\theta) \to 1.$$
\end{enumerate}
\end{theorem}

\new{The theorem is essentially saying the following. Note first that there are two separate cases. If maximizers of $L$ exist, then, assuming a ``positive gap'' condition and a ``boundedness'' condition, beliefs converge to those maximizers. But if maximizers of $L$ do not exist (for example, if $L(\theta)$ increases to but is never equal to a positive number), then all that we can say is that the beliefs concentrate at points where $L$ are near the supremum, with a boundedness condition similar to the first case.}

\new{We focus on the first case, i.e. when maximizers of $L$ exist, and assume a unique maximizer. The interesting question is whether these ``positive gap'' and ``boundedness'' conditions, which} can be traced back to \new{the conditions in Proposition \ref{prop:nearessup} and then} the $L(\theta)g_i^{(0)}(\theta)^{\delta_i} \leq L(\theta')g_i^{(0)}(\theta')^{\delta_i}$ \new{``robust domination''} hypothesis of the Master Theorem \ref{thm:master}, \new{are necessary. If they are, then, as alluded to in Section \ref{subsec:master}, this shows that the robust domination} hypotheses of the Master Theorem cannot be omitted; \new{in particular, they are not just an artifact of how we proved that theorem}.
The following example shows that neither of the two hypotheses of the first case can be omitted.
\new{These examples are rather delicate. Specifically, we know that the exponent in Proposition \ref{prop:ratio_of_f_asymptotic} grows like $r^t$; but in order to construct these examples, the error term $\eps_{i,j}^{(t)}$ has to be estimated as well.
This gives rise to expressions like $4^\theta+2^{\theta+1}$.}

\begin{example}
\label{ex:countercountable}
Counterexamples to the conclusion of the first case of Theorem \ref{thm:maincountable} that satisfy only the first and second hypotheses, respectively.

\begin{enumerate}[wide]
\item \new{This example satisfies the positive gap, but not the boundedness hypothesis.} Consider the $(3,1)$-graph of Example \ref{ex:abgraph}. Let $\Theta=\Z_{\geq 0}$, $0<\alpha<1$,
$$f_*(\theta)=\begin{cases}
\alpha, & \theta=0 \\
\alpha^{4^{\theta}}, & \theta \geq 1
\end{cases},\quad
f_x^{(0)}(\theta)\propto \begin{cases}
1, & \theta=0 \\
\alpha^{4^{\theta}-2^{\theta+1}}, & \theta \geq 1
\end{cases},\quad
f_y^{(0)}(\theta)\propto \begin{cases}
1, & \theta=0 \\
\alpha^{4^{\theta}+2^{\theta+1}}, & \theta \geq 1
\end{cases}.$$
\new{The idea is to set things up in such a way that at round $t$, the parameter $\theta=t$ holds a nontrivial fraction of beliefs.} Then $L(\theta)=\alpha L(0)$ for all $\theta \geq 1$, so $0$ uniquely maximizes $L$ \new{and the positive gap condition is satisfied}.
We \new{now} compute
$$f_x^{(t)}(0)=\frac{1}{1+\sum_{\theta\geq 1} \alpha^{(2^\theta-2^t)^2-1}}\leq \frac{1}{1+\alpha^{-1}}<1$$
for all $t \geq 1$, so that \new{beliefs do not converge to a point distribution at 0} and the conclusion of the theorem does not hold.
\item \new{This example satisfies the boundedness, but not the positive gap hypothesis.} Same as the previous example but with the initial condition
$$f_*(\theta)=\begin{cases}
1, & \theta=0 \\
\alpha^{2^\theta}, & \theta \geq 1
\end{cases},\quad
f_x^{(0)}(\theta)\propto \begin{cases}
1, & \theta=0 \\
\alpha^{4^{-\theta}+2^\theta-1}, & \theta \geq 1
\end{cases},\quad
f_y^{(0)}(\theta)\propto \begin{cases}
1, & \theta=0 \\
\alpha^{4^{-\theta}+2^\theta+1}, & \theta \geq 1
\end{cases}.$$
\new{It is easily seen that the boundedness condition is satisfied.} Then $L(\theta)=\alpha^{4^{-\theta}}L(0)$ for all $\theta \geq 1$, so $0$ uniquely maximizes $L$.
Similarly, \new{beliefs do not converge to a point distribution at 0} because, for $t\geq 1$,
$$f_x^{(t)}(0)=\frac{1}{1+\sum_{\theta\geq 1} \alpha^{2^\theta+4^{t-\theta}-2^t}}\leq \frac{1}{1+\alpha}<1.$$
\new{Again, we set things up in such a way that the term with $\theta=t$ keeps the denominator large in round $t$.}
\end{enumerate}
\end{example}

\subsection{$\R^k$ Case}
\label{subsec:Rk}

\new{Finally, we deal with convergence in the most difficult case, $\Theta=\R^k$.} Convergence here depends on the topology of $\R^k$, so
we need to work with functions that are sufficiently continuous. \new{Concepts from real analysis will be freely used.}
Let $\overline{A}$ denote the closure of $A$ and $B_r(a)$ the open ball of radius $r$ around $a$.

\new{First we introduce the} notion of \emph{piecewise continuity}. \new{Exactly why this is needed can be explained as follows. If the initial beliefs are continuous, then this continuity allows us to establish convergence quite easily. However, we may also be interested in densities that are not continuous, e.g. a density that is 1 on $[0,1]$ and 0 elsewhere. But this density is actually continuous \emph{enough} that we can work with it; specifically, it is continuous on each of $(-\infty,0)$, $(0,1)$ and $(1,\infty)$. Thus we turn to the notion of piecewise continuity, which} should include most densities in practice. \new{There is no universally agreed definition of piecewise continuity, so we define our own in Definition \ref{def:piecewise}. On $\R^k$, this notion essentially means that there are disjoint open sets on each of which the function is continuous. Moreover, these open sets must ``fill up $\Theta$,'' in the sense that the portion of $\Theta$ not contained in any of these sets has measure zero. Two extra conditions are needed in Definition \ref{def:piecewise} in order to avoid pathologies.} The first condition excludes cramming open sets into small spaces, \new{e.g. if the function is continuous on each of $(1/(n+1),1/n)$, this can cause problems near 0}. The second condition excludes functions with pathological behaviors at the boundaries, \new{e.g. think of $\sin(1/x)$ near $0$.}

\begin{definition}
\label{def:piecewise}
Let $\Theta \subseteq \R^k$. A function $f:\Theta \to \R_{\geq 0}$ is \emph{piecewise continuous} if there is a countable collection of disjoint open sets $U_i \subseteq \Theta$ such that $\Theta \setminus \cup_i U_i$ has measure zero; only finitely many $U_i$'s intersect any given bounded set $V \subseteq \R^k$; and for each $i$, the restriction $f:U_i \to \R_{\geq 0}$ can be extended to a continuous function $\widetilde{f}_i: \overline{U_i} \to [0,\infty]$.
\end{definition}

\new{Note that the continuous extension mentioned in the last part of the above definition has the interval in the extended real line $[0,\infty]$ as its range. This means that we allow functions such as $1/x$ to be piecewise continuous because it has a well-defined limit at 0, even though that limit is infinite.
The next lemma shows that multiple functions play well together with regard to this definition: we can choose the open sets in the definition to be the same for all functions.}

\begin{lemma}
\label{lem:piecewisegood}
Let $\Theta \subseteq \R^k$. For piecewise continuous functions $f_1,\dots,f_N:\Theta \to \R_{\geq 0}$, there are $U_i$'s as in Definition \ref{def:piecewise} that work for all of these functions.
\end{lemma}

\begin{proof}
For each $1\leq i \leq N$, let $\big( U_j^{(i)}\big)_{j\geq 1}$ be a collection of open sets that works for $f_i$. It is a simple exercise, \new{although there are details to check,} to show that the collection $\big( U_{j_1}^{(1)} \cap \dots \cap U_{j_N}^{(N)} \big)_{j_1,\dots,j_N\geq 1}$ works for every $f_i$.
\end{proof}

We now show convergence of beliefs when $\Theta\subseteq\R^k$ and $g_i^{(0)}$ is piecewise continuous.
The following theorem can be regarded as the most significant result in this paper.
\new{Recall that we allow initial beliefs to be piecewise continuous, which should include most situations in practice.
Moreover, the parameter region $\Theta$ can be any subset of $\R^k$; e.g. $\Theta=(0,\infty)$ is allowed.
However, in the case that $\Theta$ is a proper subset of $\R^k$,} the point that the beliefs converge to may be \emph{outside} of $\Theta$, but it must be a limit point of $\Theta$. \new{This is not too surprising: if the initial beliefs are $1/x^2$ on $(0,\infty)$, then it is likely that the beliefs will converge to 0, which is actually outside of the parameter set. Nevertheless, our theorem is able to account for cases like this.}

\begin{theorem}
\label{thm:Rkmain}
Assume that $\Theta \subseteq \R^k$ and the initial condition is feasible. Let $g_i^{(0)}$ be piecewise continuous. Let $L_{\esssup}=\esssup_{\theta \in \Theta} L(\theta) \leq \infty$, and for any $M$, let
$S_M=\{\theta \in \Theta: L(\theta)>M\}$. Let $\theta_{\esssup}$ be the set of all $\theta \in \overline{\Theta}$ such that for every $M<L_{\esssup}$ and $\delta>0$, $B_\delta (\theta) \cap S_M$ has positive measure.
Suppose that there is $r>0$ such that $\esssup_{|\theta|>r} L(\theta) < L_{\esssup}$.
Moreover, there is $M'<L_{\esssup}$ such that for every \new{$M \in (M',L_{\esssup})$}, $L$ is not constant a.e. on $S_M$; $g_i^{(0)}$ is bounded a.e. on $\Theta \setminus S_M$; and $g_i^{(0)}$ is bounded below by a positive number a.e. on $S_M$.
Then $\theta_{\esssup}$ is a nonempty compact set, and for any open set $U$ containing $\theta_{\esssup}$, as $t\to\infty$,
$$\int_{\theta \in U \cap \Theta} f_i^{(t)}(\theta)\to 1.$$
In particular, if $\theta_{\esssup}$ consists of a single point, then $f_i^{(t)}$, viewed as a density on $\R^k$ that vanishes outside of $\Theta$, converges to the point distribution at that point.
\end{theorem}

\new{The various hypotheses of the theorem fit together in the following way. First, $\theta_{\esssup}$ is supposed to be the ``maximizers of $L$,'' except that there are two complications. One, we must not be distracted by sets of measure zero. Two, these maximizers may lie outside of $\Theta$; think $1/x^2$ with $\Theta = (0,\infty)$.
The definition in the theorem that captures both of these aspects is the following: the ``maximizers'' are $\theta$ such that in small balls around $\theta$, the set where $L$ is near its essential supremum has positive measure. Then to show that beliefs converge to a point distribution at $\theta_{\esssup}$ (if it consists of a single point), we require $L$ to decay away at infinity, in the sense that the set where $L$ is near its essential supremum is bounded; this allows us to use compactness arguments.
Then several hypotheses on nonconstancy and boundedness from above and below are required in order to use case (1) of Proposition \ref{prop:nearessup}.
The conclusion is that beliefs concentrate in any open set containing $\theta_{\esssup}$, and this shows convergence in the case where $\theta_{\esssup}$ consists of a single point.}

\begin{shownto}{abb}
\begin{proofsketch}
If $\esssup_{|\theta|>r} L(\theta) <M< L_{\esssup}$, then $S_M \setminus B_r(0)$ has measure zero.
From this, we can show that $\theta_{\esssup} \subseteq \overline{B_r(0)}$ is bounded.
By Lemma \ref{lem:piecewisegood}, we can take $(U_i)_{i\geq 1}$ that works for all $g_i^{(0)}$. Assume that $U_1,\dots,U_n$ are the only ones whose closures intersect $\overline{B_r(0)}$. By standard arguments, we can prove that $\theta_{\esssup}=\cup_{i=1}^n \{\theta \in \overline{U_i}: \widetilde{L}_i(\theta)=L_{\esssup}\}$, where $\widetilde{L}_i$ is the continuous extension of $L$ on $\overline{U_i}$.
It follows that $\theta_{\esssup}$ is compact.

Proposition \ref{prop:nearessup} implies that $\int_{\theta \in S_M} f_i^{(t)}(\theta) \to 1$ for any $M'<M<L_{\esssup}$.
To prove the theorem, it suffices to show that there is $M'<M<L_{\esssup}$ such that $S_M\setminus U$ has measure zero, which for large $M$ is equivalent to
$(S_M \cap U_i \cap B_r(0))\setminus U$ having measure zero for all $1\leq i\leq n$.
To prove this, show that the existence of an increasing sequence $M_j \to L_{\esssup}$ such that $(S_{M_j} \cap U_i \cap B_r(0))\setminus U$ has positive measure leads to a contradiction.
\end{proofsketch}
\end{shownto}

\begin{shownto}{full}
\begin{proof}
\new{The proof is unfortunately rather technical.
The first step is to show that $\theta_{\esssup}$ is bounded; for this, we use that $L$ decays away at infinity. Take $r$ such that $\esssup_{|\theta|>r}<L_{\esssup}$, which exists by assumption.
We will show that $\theta_{\esssup} \subseteq \overline{B_r(0)}$.}
If $\esssup_{|\theta|>r} L(\theta) <M< L_{\esssup}$, then $S_M$ has positive measure and $S_M \setminus B_r(0)$ has measure zero.
For any $\theta \in \theta_{\esssup}$ and $\delta >0$, $B_\delta(\theta) \cap S_M$ has positive measure \new{by definition}, so \new{it follows that} $B_\delta(\theta) \cap B_r(0)$ has positive measure, so $\theta \in \overline{B_r(0)}$.
We conclude that $\theta_{\esssup} \subseteq \overline{B_r(0)}$ is bounded.

\new{The second step is to show that $\theta_{\esssup}$ is compact. We have already shown it to be bounded, so it remains to show that it is closed.}
By Lemma \ref{lem:piecewisegood}, we can take $(U_i)_{i\geq 1}$ that works for all $g_i^{(0)}$.
\new{Then note that $L$ is piecewise continuous with the $U_i$'s working for it as well.} Assume that $U_1,\dots,U_n$ are the only ones whose closures intersect $\overline{B_r(0)}$; \new{here we are using our definition of piecewise continuity to only consider $U_i$'s that are relevant.} We \new{now} claim \new{the following}: $$\theta_{\esssup}=\cup_{i=1}^n \{\theta \in \overline{U_i}: \widetilde{L}_i(\theta)=L_{\esssup}\},$$ where $\widetilde{L}_i$ is the continuous extension of $L$ on $\overline{U_i}$.
\new{In other words, $\theta_{\esssup}$ really consists of ``maximizers:'' $\theta$'s where $L$ are equal to $L_{\esssup}$, but we need to first extend $L$ onto the boundaries of $U_i$'s.}
Let $\theta \in \overline{U_i}$ be such that $\widetilde{L}_i(\theta)=L_{\esssup}$. \new{We will show that $ \theta \in\theta_{\esssup}$ just by unpacking definitions.} For every $M<L_{\esssup}$, there is $\delta'>0$ such that $B_{\delta'}(\theta) \cap U_i \subseteq S_M$ by continuity. So for any $\delta>0$, $B_{\delta}(\theta) \cap S_M \supseteq B_{\delta}(\theta) \cap B_{\delta'}(\theta) \cap U_i$ has positive measure because the latter is a nonempty open set. Hence $\theta \in \theta_{\esssup}$. Conversely, let $\theta \in \theta_{\esssup}$. Then for every $M<L_{\esssup}$ and $\delta >0$, $B_{\delta}(\theta) \cap S_M$ \new{has positive measure,} so it intersects some $U_i$. For $\delta<1$, $B_{\delta}(\theta) \subseteq B_1(\theta)$, so there are finitely many possible choices of $U_i$'s \new{ by the definition of piecewise continuity}. \new{Now} pick sequences $M_j\to L_{\esssup}$ and $\delta_j \to 0$; \new{then} some choice of $U_i$ must repeat infinitely many times. Thus
$\theta \in \overline{U_i}$ for \new{this} $i$. 
Because $\theta \in \overline{B_r(0)}$, we have $1\leq i\leq n$.
\new{Finally,} since we let $M_j \to L_{\esssup}$,
\new{the fact that $B_{\delta_j}(\theta) \cap S_{M_j} \cap U_i \neq \varnothing$ implies that} $\widetilde{L}_i(\theta) \geq L_{\esssup}$. \new{But} if $\widetilde{L}_i(\theta) > L_{\esssup}$, then there is a nonempty open set on which $L>L_{\esssup}$, a contradiction.
So $\widetilde{L}_i(\theta) = L_{\esssup}$, proving the claim. It follows from the claim that $\theta_{\esssup}$ is closed, so because it is bounded, it is compact.

\new{We now show that whenever $U$ is an open set containing $\theta_{\esssup}$, beliefs on $U \cap \Theta$ tend to 1. Our strategy is the following.}
Proposition \ref{prop:nearessup} implies that for any $M'<M<L_{\esssup}$, as $t\to\infty$,
$\int_{\theta \in S_M} f_i^{(t)}(\theta) \to 1$.
\new{So} it suffices to show that there is $M'<M<L_{\esssup}$ such that $S_M\setminus U$ has measure zero. \new{To this end, note that} for large $M$, this is equivalent to $(S_M\cap B_r(0))\setminus U$ having measure zero, which \new{in turn} is equivalent to $(S_M \cap U_i \cap B_r(0))\setminus U$ having measure zero for all $1\leq i\leq n$.
Suppose the contrary. \new{Then, for some $i$,} there is an increasing sequence $M_j \to L_{\esssup}$ such that $(S_{M_j} \cap U_i \cap B_r(0))\setminus U$ has positive measure. Pick $\theta_j \in (S_{M_j} \cap U_i \cap B_r(0))\setminus U$. Then $\theta_j$ is a bounded sequence in $U_i$ such that $L(\theta_j) \to L_{\esssup}$. We can pass to a subsequence and assume that $\theta_j \to \theta \in \overline{U_i}$. Then $\theta \in \theta_{\esssup}$ \new{by the claim above}, so $\theta_j \in U$ for large $j$, a contradiction.
Therefore the desired integral converges to 1.
\new{This argument can also be modified to show that $\theta_{\esssup}$ is nonempty: take $U=\varnothing$ in the above argument, note that $S_M\setminus U = S_M$ has positive measure, and run the argument to construct $\theta \in \theta_{\esssup}$ as above. Alternatively, the fact that beliefs on $U \cap \Theta$ tend to 1 for any open $U$ containing $\theta_{\esssup}$ already implies that $\theta_{\esssup} \neq \varnothing$, because otherwise the convergence should be true for $U =\varnothing$ as well.}

\new{Finally,} to see that the distribution converges to the point distribution at $\theta$ if $\theta_{\esssup}=\{\theta\}$, take $U=B_{\delta}(\theta)$ for small $\delta >0$.
\end{proof}
\end{shownto}

\new{Lastly,} we present a simplified version of Theorem \ref{thm:Rkmain} when $\Theta=\R^k$ and $g_i^{(0)}$ is bounded and continuous.
\new{This version can deal with densities that are continuous on the whole of $\R^k$ with some mild decay conditions, and it is applicable to e.g. Gaussians and many similar densities. To prove this corollary, we directly apply Theorem \ref{thm:Rkmain}, where}
Corollary \ref{cor:suffcientfeas} is used to prove that the initial condition is feasible.

\begin{cor}
\label{cor:Rkcontinuous}
Assume that $\Theta = \R^k$. Let $g_i^{(0)}$ be bounded continuous functions. Suppose that there is a bounded set $S$ such that $\sup_{\theta \notin S} L(\theta) < \sup_{\theta \in \R^k} L(\theta)$.
Then the initial condition is feasible, $L$ attains a maximum at the set of points $\theta_{\max}\neq\varnothing$, and for any open set $U$ containing $\theta_{\max}$, as $t\to\infty$,
$$\int_{\theta \in U} f_i^{(t)}(\theta)\to 1.$$
In particular, if $\theta_{\max}$ consists of a single point, then $f_i^{(t)}$ converges to the point distribution at that point.
\end{cor}

\new{Another (further simplified) version of Corollary \ref{cor:Rkcontinuous} that should be useful in practice is Theorem \ref{thm:mainsimplified} in the introduction.}

\section{Applications}
\label{sec:apps}

\new{We now seek to apply our theoretical results to model real-world situations. One example is given for each of the cases where $\Theta$ is finite, infinite discrete, and continuous to demonstrate the use of our theorems.}

\subsection{Binary Beliefs}
Suppose that the underlying state of the world is binary, $\Theta = \{0,1\}$. For example, agents are trying to determine the truth value of a statement.
\new{At first, agent $i$ believes that the statement holds with probability $x_i$.}
By Theorem \ref{thm:mainfinite}, we can conclude the following.

\begin{prop}
Let the initial belief of agent $i$ be Bernoulli distributed as $\operatorname{Bern}(x_i)$, $0<x_i<1$, and the common prior be $\operatorname{Bern}(1/2)$. Then the initial condition is feasible. If $\prod_{i=1}^{n} x_i^{v_i} > \prod_{i=1}^{n} (1-x_i)^{v_i}$, then $f_i^{(t)}$ converges to the point distribution at 1; if $\prod_{i=1}^{n} x_i^{v_i} < \prod_{i=1}^{n} (1-x_i)^{v_i}$, then $f_i^{(t)}$ converges to the point distribution at 0.
\end{prop}

\new{Note that, except in the borderline case $\prod_{i=1}^{n} x_i^{v_i} = \prod_{i=1}^{n} (1-x_i)^{v_i}$, agents will reach a consensus on whether the statement is true or false.}
This process can be viewed as if agents cast votes on what everyone should believe: everyone will unanimously agree that the statement is true if
\new{$$\sum_{i=1}^n v_i \log \paren{\frac{x_i}{1-x_i}} > 0,$$
and that the statement is false if this expression is negative.
Thus it is as if}
agent $i$ casts votes of $\log (x_i/(1-x_i))$ in support of the statement and gets $v_i$ votes.
The expression $\log (x_i/(1-x_i))$, called the \emph{log-odds}, transforms probability in $(0,1)$ to the real line.
\new{Moreover, we can see that} the strength of an agent's vote \new{indeed} comes from her confidence and her centrality.

\new{A similar analysis holds when agents are trying to decide between a finite number of alternatives. In this case, if agent $i$ believes in alternative $j$ with probability $x_{i,j}$, then it is as if she casts votes of $\log x_{i,j}$ in support of alternative $j$. The alternative with the most votes becomes the consensus.}

\subsection{Poisson Beliefs}

\new{Now let us consider the case where there are an infinite number of alternatives, e.g. when} the underlying state of the world is the number of events in a time period. The initial beliefs can be modeled by Poisson distributions.

\begin{prop}\label{prop:pois}
Assume that the initial belief of agent $i$ is $\operatorname{Pois}(\lambda_i)$ and the common prior is a flat prior, then the initial condition is feasible. Let $\lambda^* = \prod_{i=1}^{N} \lambda_i^{v_i}$. If $\lambda^*$ is not an integer, then $f_i^{(t)}$ converges to the point distribution at $\lfloor \lambda^* \rfloor$. If $\lambda^*$ is an integer, then $f_i^{(t)}$ lies in $\{\lambda^*-1,\lambda^*\}$ with probability going to 1 as $t \to \infty$.
\end{prop}

\begin{shownto}{abb}
\begin{proofsketch}
Compute $L(\theta) \propto (\lambda^*)^{\theta}/\theta!$ and apply Theorem \ref{thm:maincountable}. 
\end{proofsketch}
\end{shownto}

\begin{shownto}{full}
\begin{proof}
\new{Since $g_i^{(0)}$ is bounded and the initial condition is nondegenerate}, the initial condition is feasible by Corollary \ref{cor:suffcientfeas}. \new{It is easy to compute that}
$$L(\theta) \propto \frac{(\lambda^*)^{\theta}}{\theta!}.$$
By direct computation, $L$ has a unique maximum at $\lfloor \lambda^* \rfloor$ if $\lambda^*$ is not an integer, and two maxima $\{\lambda^*-1,\lambda^*\}$ if $\lambda^*$ is an integer. \new{Since $L(\theta) \to 0$ as $\theta \to \infty$ and $g_i^{(0)}$ is bounded, we can apply} Theorem \ref{thm:maincountable} \new{to show convergence}.
\end{proof}
\end{shownto}

The Poisson parameter $\lambda$ is the mean of $\operatorname{Pois}(\lambda)$ and is related to the average event rate. Our result shows that \new{agents will reach a consensus at $\lambda^*$, the \emph{weighted geometric mean} of the rates corresponding to agents' beliefs, where again the weight of each agent is her centrality. Note that our model indicates that Poisson} rates should be combined multiplicatively \new{rather than additively. This accords with intuition:} if two connected people believe that there will be $2$ and $1000$ floods on average next year, respectively, then \new{Proposition \ref{prop:pois} says that} they will come to believe that there will be $\lfloor \sqrt{2000} \rfloor = 44$
floods. In contrast, combining rates additively as in DeGroot leads to the counterintuitive answer of 501
floods. \new{Thus, even though each agent carries essentially one piece of information: the event rate, by modeling the beliefs according to the actual underlying parameter space, our model leads us to a natural consensus.
A similar situation is} the Gaussian belief case in the next section where beliefs converge to the weighted arithmetic mean and agent $i$ also has weight $v_i$.

\subsection{Gaussian Beliefs}
\label{subsec:gaussian}

\new{Finally, we consider the interesting case of each agent carrying \emph{two} pieces of information, her belief and her confidence.} Suppose that the underlying state of the world is $\Theta = \R$ with the flat prior $f_*\equiv 1$.
Each agent observes the true state of the world $\theta^*$ with Gaussian noise:
agent $i$'s signal $\mu_i$ is drawn from the normal distribution $\N(\theta^*,1/\tau_i)$, where $\tau_i$ is the \emph{precision} known to agent $i$.
\new{Note that this is an example of an underlying scenario explained in Section \ref{subsec:modelupdates}, except that the prior in this case is improper.}
Then agent $i$'s initial belief for $\theta^*$ \new{(conditional on the signal $\mu_i$) can be computed to be} $\N(\mu_i,1/\tau_i)$.

\begin{prop}\label{prop:gaussian}
Assume that the initial belief of agent $i$ is normally distributed as $\N(\mu_i, 1/\tau_i)$ and the common prior is a flat prior, then this initial condition is feasible and $f_i^{(t)}$ converges to the point distribution at
$$\theta_{\max} = \sum_{i=1}^{N} c_i \mu_i,\quad c_i =\frac{v_i \tau_i}{ \sum_{j=1}^{N} v_j \tau_j}.$$
\end{prop}

\begin{shownto}{abb}
\begin{proofsketch}
The weighted likelihood function
is proportional to a Gaussian with the maximum at $\theta_{\max}$.
Apply Corollary \ref{cor:Rkcontinuous}.
\end{proofsketch}
\end{shownto}

\begin{shownto}{full}
\begin{proof}
The weighted likelihood function is
$$L(\theta) \propto \exp \left[ - \frac{1}{2}\sum_{i=1}^{N} v_i \tau_i (\theta- \mu_i)^2 \right],$$
\new{and we can see that this} is proportional to a Gaussian with the maximum at $\theta_{\max}$.
\new{Since $g_i^{(0)}$ are bounded and continuous and $L$ decays quickly at infinity, the convergence} follows from Corollary \ref{cor:Rkcontinuous}.
\end{proof}
\end{shownto}

If we interpret the initial belief $\N(\mu_i,1/\tau_i)$ as a scalar belief $\mu_i$ with confidence $\tau_i$, then the consensus belief \new{$\theta_{\max}$} is the weighted average of agents' signals \new{$\mu_i$} with weights \new{$c_i$} proportional to her centrality $v_i$ and her precision $\tau_i$.
Therefore, an agent who is more \text{centrally located} and has a more \text{informative signal} has more influence over the consensus belief. This phenomenon cannot be captured by the standard DeGroot model which has no notion of quality of signals.

We can evaluate quality of learning by the variance of the consensus. \new{This works as follows.} The consensus \new{$\theta_{\max}$} is a function of the signals \new{$\mu_i$}, and the signals are random, so the consensus can be viewed as a random variable. \new{Specifically,} it is an unbiased estimator of the true state $\theta^*$, and the lower its variance, the higher the quality of learning. Since $\mu_i \sim \N(\theta^*,1/\tau_i)$ are independent, $\theta_{\max} = \sum_i c_i \mu_i \sim \N(\theta^*, S)$ \new{remains a Gaussian, with variance}
$$S = \frac{\sum_i v_i^2 \tau_i}{(\sum_i v_i \tau_i)^2}.$$ \new{We now investigate the effect of increasing the precision $\tau_k$ on the variance $S$. We find that, all other variables being fixed, the result depends on the threshold} $V_k = 2 \sum_{i \neq k} v_i^2 \tau_i / \sum_{i \neq k} v_i \tau_i$ \new{of the centrality $v_k$}. By direct computation, if $v_k \leq V_k$,
then increasing $\tau_k$ always decreases $S$, while if $v_k>V_k$, then increasing $\tau_k$ decreases $S$ only if $\tau_k \geq (v_k-V_k)\sum_{i \neq k} v_i \tau_i/v_k^2$ and \emph{increases} $S$ otherwise.
\new{Thus, increasing the precision $\tau_k$ usually decreases the variance and improves the quality of learning, \emph{except} when the centrality $v_k$ is high and the precision $\tau_k$ is low.} This result has an implication for social planners who want to encourage good learning. The social planner cannot control the network structure $v_k$ but she might be able to control the signal precision $\tau_k$.
If an agent is central but relatively uninformed, increasing the precision of that agent can reduce learning quality. \new{An intuitive explanation might be that} the agent now has some confidence to use her centrality to impose her less informed opinion on the consensus. Only when that central agent is sufficiently informed will additional precision improve learning quality. Therefore, if a social planner wants to seed centrally located agents (opinion leaders) with new technologies, she must invest enough resources to make those agents well informed; otherwise, the effort can backfire. 

Note that we assume independence in the agents' update rule, not in the initial signals themselves. Therefore, our framework can deal with arbitrarily correlated initial signals. In the case of Gaussian beliefs, if we assume that the correlation of initial signals of agents $i$ and $j$ is $\rho_{ij}$, then the consensus $\theta_{\max}$ is still unbiased with variance $$\frac{\sum_i v_i^2 \tau_i + 2 \sum_{i<j} v_i v_j \sqrt{\tau_i \tau_j} \rho_{ij}}{(\sum_i v_i \tau_i)^2}.$$ This variance is increasing in every correlation $\rho_{ij}$, which indicates that higher correlations in initial signals reduce learning equality, as expected.

\new{Another point of interest is that our consensus does not depend on the geometry of the network beyond the eigenvector centralities; this is true for all of our analysis from Section \ref{sec:concentrate} onwards and not only for this particular Gaussian case. Let us contrast this with the work of} \citet{BBCM}, which also incorporates quality of signals into DeGroot learning. In their model, the quality of signals is binary (informed/uninformed), and each agent takes the average of signals of only her informed neighbors.
They identify the \textit{clustered seeding} problem: if the initially informed agents are closely connected, some agents will ``block'' other agents, reducing the influence of the blocked agents.
To illustrate this point, consider \new{a social network in the shape of} an undirected cycle of
length $N$ and assume that only three adjacent agents $1,2,3$ are initially informed with the same precision. In our model, the three agents' signals have equal weights: $c_1=c_2=c_3=1/3$. But in their model, agent 2 is blocked by agents 1 and 3: $c_1=c_3 \to 1/2$ and $c_2 \to 0$ as the number of agents $N \to \infty$. \new{Note that all three agents do have the same centralities and the same precisions, yet the geometry of the network still allows them to block each other. Why this occurs may be explained as follows: as the ``informed segment'' of the network expands, an agent who has just become newly informed will get the opinion of agent 1 or 3 first. Surely she will not be very confident because she has just heard of one opinion, but in the model of \citet{BBCM}, her opinion counts as much as everyone else, and so the opinions of agents 1 and 3 get injected back into the informed segment and reinforce themselves to an inappropriate degree. In contrast, our model circumvents this problem by allowing agents to express their confidence. Thus, a newly informed agent is not very confident and her opinion counts less; and it turns out that this alone is enough to offset the blocking disadvantage and allows the opinion of agent 2 to come through in the consensus.}
 We conclude that the clustered seeding problem occurs when agents can communicate whether they are informed but cannot communicate precisions of their signals.
When agents can communicate precisions of signals, even if beliefs are naively updated, the problem disappears.
\new{The conclusion is that it is of vital importance that agents be able to communicate opinions freely and expressively, and not just state their beliefs, in order for blocking to not occur.}

\section*{Acknowledgments}
We thank Ben Golub for insightful comments in the early stage of this project. We are also grateful to Nicole Immorlica and Evan Sadler for suggestions that improved the paper.

\bibliographystyle{ACM-Reference-Format}
\bibliography{references}

\begin{shownto}{abb}
\newpage
\appendix
\section{Full Proofs and Details of Examples}

\subsection*{Proof of Proposition \ref{prop:justifyrule}}

By definition, the updated belief is
$$f(\theta)=p\paren{\theta|(X_j)_{j \in N(i)}}
= \frac{p(\theta)p\paren{(X_j)_{j \in N(i)} |\theta} }{\int_{\theta' \in \Theta} p(\theta') p\paren{(X_j)_{j \in N(i)} |\theta'} }
= \frac{p(\theta) \prod_{j \in N(i)} p\paren{X_j |\theta} }{\int_{\theta' \in \Theta} p(\theta') \prod_{j \in N(i)} p\paren{X_j |\theta'} },$$
where the second equality is Bayes' rule and the last equality holds because $X_j$'s are conditionally independent given any $\theta$.
The integral in the denominator is nonzero and finite due to the marginalizability assumption.
Notice that $p(\theta)=f_*(\theta)$ and
$$p(X_j|\theta) = \frac{p(\theta|X_j)}{p(\theta)} p(X_j) = \frac{f_j(\theta)}{f_*(\theta)}p(X_j).$$
Substituting this into our first equation
gives the desired formula.

\subsection*{Proof of Lemma \ref{lem:primitive}}

By strong connectedness, for every $i$ and $j$, there is a path $P_{ij}$ from $i$ to $j$. Let $k$ be the maximum length of all $P_{ij}$'s. We can extend a path to arbitrary lengths by using self-loops, so for any $i$ and $j$, there is a path from $i$ to $j$ of length $k$. Because the $(i,j)$ entry of $A^k$ is the number of paths from $i$ to $j$ of length $k$ (see e.g. Lemma \ref{lem:path_recurrence}), $A^k$ has positive entries. Therefore $A$ is primitive.

\subsection*{Proof of Lemma \ref{lem:spectral_radius_greater_than_one}}

Let $r=\rho(A)$. Lemma \ref{lem:primitive} implies that $A$ is primitive, so by Theorem \ref{thm:perron}, $r$ is an eigenvalue of $A$,
and there is an eigenvector $v$ associated with $r$ with positive entries. Pick an arbitrary edge $(i,j)$ with $i \neq j$. Then
$$rv_i=\sum_{i'=1}^N A_{ii'}v_{i'} \geq v_i + v_j > v_i,$$
implying that $r>1$.

\subsection*{Proof of Lemma \ref{lem:newholder}}

Apply H\"older's inequality (Theorem \ref{thm:holder}) to the functions $f_*f_i^p$.

\subsection*{Proof of Lemma \ref{lem:interpolationLp}}

First assume that $r < \infty$. There are $p_1,p_2 \in (0,1)$ such that $p_1+p_2=1$ and $p_1p+p_2r=q$. If $f \in L^p_* \cap L^r_*$, then $f^p,f^r \in L^1_*$.
By Lemma \ref{lem:newholder}, $f^{p_1p+p_2r}=f^q \in L^1_*$, so that $f \in L^q_*$.

We directly prove the case where $r=\infty$. Let $M=\norm{f}_{\infty,*}$. Then $f \leq M$ almost everywhere, so
$$\int_{\theta \in \Theta} f_*f^q \leq
M^{q-p}\int_{\theta \in \Theta} f_*f^p<\infty.$$
Therefore $f \in L^q_*$.

\subsection*{Proof of Lemma \ref{lem:path_recurrence}}

The recurrence follows by considering that a path from $i$ to $j$ passes through exactly one node $k \in N(j)$ just before arriving at $j$. Let $P^{(t)}$ be the matrix whose $(i,j)$ entry is $P_{ij}^{(t)}$. The recurrence can be written in the form $P^{(t+1)}=P^{(t)}A$, with the initial condition $P^{(0)}=I$. It follows by induction that $P^{(t)}=A^t$.

\subsection*{Proof of Proposition \ref{prop:update_t_g}}

We induct on $k$. Any initial condition is 0-feasible. For any $i$, $\prod_{j=1}^N (g_j^{(0)})^{P_{ji}^{(0)}} = g_i^{(0)} \in L^1_*$ with norm 1, so the formula holds. This proves the base case $k=0$.

Assume that the proposition holds for $k$. For a $(k+1)$-feasible initial condition, the updated beliefs given by equation \eqref{eq:updateg}
$$g_i^{(k+1)}(\theta) = \frac{\prod_{j \in N(i)}{g_j^{(k)}(\theta)}}
{\int_{\theta' \in \Theta} f_*(\theta') \prod_{j \in N(i)}{g_j^{(k)}(\theta')}}$$
are well-defined.
By the induction hypothesis, this equals
\begin{equation}
\label{eq:updateformula}
g_i^{(k+1)}(\theta) =
\frac{\prod_{j \in N(i)} \prod_{\ell=1}^{N} g_\ell^{(0)}(\theta)^{P_{\ell j}^{(k)}} }{\int_{\theta' \in \Theta} f_*(\theta') \prod_{j \in N(i)} \prod_{\ell=1}^{N} g_\ell^{(0)}(\theta')^{P_{\ell j}^{(k)}} } = 
\frac{\prod_{\ell=1}^{N} g_\ell^{(0)}(\theta)^{P_{\ell i}^{(k+1)}} }{\int_{\theta' \in \Theta} f_*(\theta') \prod_{\ell=1}^{N} g_\ell^{(0)}(\theta')^{P_{\ell i}^{(k+1)}} },
\end{equation}
where the second equality holds because the exponent of $g_\ell^{(0)}$ is $\sum_{j \in N(i)} P_{\ell j}^{(k)}=P_{\ell i}^{(k+1)}$ by Lemma \ref{lem:path_recurrence}.
So the formula holds. The denominator of the right-hand expression shows that $\prod_{j=1}^N (g_j^{(0)})^{P_{ji}^{(k+1)}} \in L^1_*$ with nonzero norm.

Conversely, suppose that $\prod_{j=1}^N (g_j^{(0)})^{P_{ji}^{(t)}} \in L^1_*$ with nonzero norm for all $t\leq k+1$. By the induction hypothesis, the initial condition is $k$-feasible. The right-hand expression of equation \eqref{eq:updateformula} is well-defined, so by the same argument as above, it can be written as the expression for $g_i^{(k+1)}(\theta)$ in the update rule. Therefore $g_i^{(k+1)}$ is well-defined and the initial condition is $(k+1)$-feasible.

\subsection*{Proof of Proposition \ref{prop:underlyinggood}}

Recall that agent $i$'s initial belief is $f_i^{(0)}(\theta)=p(\theta|X_i)$, where $X_i$ is her signal.
Denote by $(X_i^{n_i})_{i=1}^{N}$ the ordered tuple of $n_i$ copies of $X_i$ for $1\leq i \leq N$.
By the marginalizability assumption,
\begin{align*}
p\big( (X_i^{n_i})_{i=1}^{N} \big) &= \int_{\theta \in \Theta} f_*(\theta) p\big( (X_i^{n_i})_{i=1}^{N} | \theta\big) = \int_{\theta \in \Theta} f_*(\theta) \prod_{i=1}^{N} p\paren{X_i| \theta}^{n_i} \\
&= \int_{\theta \in \Theta} f_*(\theta) \prod_{i=1}^{N} \bigg(\frac{f_i^{(0)}(\theta) p(X_i)  }{f_*(\theta)}\bigg)^{n_i} =
\bigg( \prod_{i=1}^{N} p(X_i)^{n_i}\bigg)
\int_{\theta \in \Theta} f_* \prod_{i=1}^{N} (g_i^{(0)})^{n_i}
\end{align*}
is nonzero and finite. 
So $\prod_{i=1}^{N} (g_i^{(0)})^{n_i} \in L_*^{1}$ with nonzero norm.

\subsection*{Proof of Proposition \ref{prop:underlyingfeasible}}

Propositions \ref{prop:update_t_g} and \ref{prop:underlyinggood} imply the first statement. As in Proposition \ref{prop:justifyrule}, we compute
\begin{align*}\label{eqn:updated_belief_repeated_signals}
p\big(\theta|I_i^{(t)}\big) =
p\bigg( \theta \Big| \Big( X_j^{P_{ji}^{(t)}} \Big)_{j=1}^N \bigg)
= \frac{f_*(\theta) \prod_{j=1}^{N} g_j^{(0)}(\theta)^{P_{ji}^{(t)}} }{ 
\int_{\theta' \in \Theta} f_*(\theta') \prod_{j=1}^{N} g_j^{(0)}(\theta')^{P_{ji}^{(t)}}
    },
\end{align*}
which equals $f_i^{(t)}(\theta)$ by Proposition \ref{prop:update_t_g}.

\subsection*{Proof of Proposition \ref{prop:diagramworks}}

The region $\mc{U}\neq \varnothing$ because an underlying scenario clearly exists for any $\Theta$. Examples of initial conditions in other regions are given in the following table, where $c,c_1,c_2$ are normalizing constants.

\begin{center}
\begin{tabular}{| c | c | c | c | c |}
\hline
Region & $f_*(\theta)$ & $f_1^{(0)}(\theta)$ & $f_2^{(0)}(\theta)$ \\ \hline
$\mc{P} \setminus \mc{F}_1$ & $\theta^{-3}$ & $c\theta^{-2}$ & $c\theta^{-2}$ \\ \hline
$\mc{IP} \setminus \mc{F}_1$ & $\theta^{-3}$ for odd $\theta$, $1$ for even $\theta$ & $c\theta^{-2}$ & $c\theta^{-2}$ \\ \hline
$\paren{\mc{F}_k \cap \mc{P}} \setminus \paren{\mc{F}_{k+1}\cap \mc{P}}$ & $\theta^{-(2^{k+1}+1)}$ & $c\theta^{-2^{k+1}}$ & $c\theta^{-2^{k+1}}$ \\
\hline
$\paren{\mc{F}_k \cap \mc{IP}} \setminus \paren{\mc{F}_{k+1} \cap \mc{IP}}$ & $\theta^{-(2^{k+1}+1)}$ for odd $\theta$, $1$ for even $\theta$ & $c\theta^{-2^{k+1}}$ & $c\theta^{-2^{k+1}}$ \\ \hline
$\paren{\mc{F}\cap\mc{P}} \setminus \mc{U}$ & $\theta^{-3}$ & $c_1 \theta^{-2}$ & $c_2 \theta^{-4}$ \\ \hline
$\mc{F}\cap \mc{IP}$ & $1$ & $c\theta^{-2}$ & $c\theta^{-2}$ \\ \hline
\end{tabular}
\end{center}

An initial condition is in $\mc{P}$ if $\int_{\theta} f_*<\infty$ and is in $\mc{IP}$ otherwise.
For any $t\geq 1$, there are $2^{t-1}$ paths from any node $i$ to any node $j$. By Proposition \ref{prop:update_t_g}, an initial condition is in $\mc{F}_k$ if and only if
$$ \int_{\theta} f_* (g_1^{(0)})^{2^{t-1}} (g_2^{(0)})^{2^{t-1}}$$
is nonzero and finite for all $1 \leq t\leq k$, and it is in $\mc{F}$ if and only if this is true for all $t \geq 1$.
These checks are straightforward using the fact that $\sum_{\theta=1}^{\infty} \theta^{-p}$ is finite for $p>1$ and infinite for $p \leq 1$. 
To check that the example for $\paren{\mc{F}\cap\mc{P}} \setminus \mc{U}$ does not have an underlying scenario, use Proposition \ref{prop:underlyinggood} and the fact that
$$\int_{\theta} f_* (g_1^{(0)})^2 = c_1^2 \int_{\theta} \theta^{-1}=\infty.$$

\subsection*{Proof of Proposition \ref{prop:nondeg}}

Pick any node $i$ and pick $k$ such that there is a path from any $j$ to $i$ of length $k$. Let $n_j=P_{ji}^{(k)}>0$. Then Proposition \ref{prop:update_t_g} implies that
$(g_1^{(0)})^{n_1}\dots (g_N^{(0)})^{n_N} \in L^1_*$
with nonzero norm. If the initial condition is degenerate, then this function vanishes almost everywhere and must have zero norm. Therefore the initial condition is nondegenerate.

\subsection*{Proof of Proposition \ref{prop:sufficient_feasibility}}

By Proposition \ref{prop:update_t_g}, it suffices to prove that for any $n_1,\dots,n_N\geq 0$ that do not vanish simultaneously,
$\prod_{i=1}^{N} (g_i^{(0)})^{n_i} \in L_*^{1}$ with nonzero norm. Let $M=n_1+\dots+n_N \geq 1$.
For each $i$, there is $M \leq p\leq \infty$ such that 
$g_i^{(0)} \in L_*^p$. Because $g_i^{(0)} \in L_*^1$ by definition, by Lemma \ref{lem:interpolationLp}, $g_i^{(0)}\in L^M_*$.
Then by Lemma \ref{lem:newholder}, $\prod_{i=1}^{N} (g_i^{(0)})^{n_i/M} \in L_*^M$.
Hence $\prod_{i=1}^{N} (g_i^{(0)})^{n_i} \in L_*^1.$
This function has nonzero norm because the set $\{\theta \in \Theta: g_i^{(0)}(\theta)>0 \text{ for all } i\}$ does not have measure zero.

\subsection*{Proof of Proposition \ref{prop:suff_feas_weakest}}

Pick any $f \in S$ and $1 \leq p <\infty$. We must show that $f \in L_*^p$. If $f$ vanishes almost everywhere, then clearly $f \in L_*^p$. So suppose otherwise. An initial condition with $g_i^{(0)}=f$ is well-defined because $f \in L^1_*$, and it is nondegenerate, so it is feasible by hypothesis.

By Proposition \ref{prop:update_t_g}, the function
$f^{P_{1i}^{(t)}+\dots+P_{Ni}^{(t)}} \in L^1_*$ for any $i$ and $t$.
By Lemma \ref{lem:path_recurrence}, these $P_{ji}^{(t)}$'s are entries of $A^t$, which go to infinity as $t\to\infty$ by Theorem \ref{thm:perron} and Lemma \ref{lem:spectral_radius_greater_than_one}.
So we obtain a sequence $p_t \to \infty$ as $t\to\infty$ such that $f \in L^{p_t}_*$. By Lemma \ref{lem:interpolationLp}, $f \in L^p$.

\subsection*{Proof of Corollary \ref{cor:suffcientfeas}}

By definition, $g_i^{(0)} \in L_*^{1}$, so by Lemma \ref{lem:interpolationLp}, $g_i^{(0)} \in L_*^{p}$ for any $1 \leq p < \infty$. Now apply Proposition \ref{prop:sufficient_feasibility}.

\subsection*{Proof of Proposition \ref{prop:LLp}}

By Proposition \ref{prop:update_t_g}, $\prod_{j=1}^N (g_j^{(0)})^{P_{ji}^{(t)}} \in L^1_*$ for any $i$ and $t$.
Recall that $\sum_i v_i=1$. By Lemma \ref{lem:newholder}, we can take the weighted geometric mean for all $i$ weighted by $v_i$ to get
$$\prod_{j=1}^N (g_j^{(0)})^{\sum_{i=1}^N v_iP_{ji}^{(t)}} \in L^1_*.$$
By Lemma \ref{lem:path_recurrence}, $P_{ji}^{(t)}$ is the $(j,i)$ entry of $A^t$, so the quantity in the exponent is the $j$th entry of $A^tv$. Because $Av=rv$ where $r=\rho(A)$, this quantity is $r^t v_j$.
We conclude that $L \in L^{r^t}_*$ for all $t\geq 0$.
Because $r>1$ by Lemma \ref{lem:spectral_radius_greater_than_one}, Lemma \ref{lem:interpolationLp} implies that $L \in L^p_*$ for all $1 \leq p < \infty$. Its norms are nonzero because the initial condition is nondegenerate (Prop. \ref{prop:nondeg}).

\subsection*{Proof of Lemma \ref{lem:Lzeroornot}}

If $L(\theta)>0$, then $f_i^{(0)}(\theta)>0$ for all $i$. The desired result follows from Proposition \ref{prop:update_t_g}.
Let $T$ be large enough so that $P_{ij}^{(t)}>0$ for all $i,j$ and $t\geq T$.
If $L(\theta)=0$, then $f_j^{(0)}(\theta)=0$ for some $j$.  Then $f_i^{(t)}(\theta)=0$ for all $i$ and $t\geq T$ by Proposition \ref{prop:update_t_g}.

\subsection*{Proof of Proposition \ref{prop:ratio_of_f_asymptotic}}

Note that $f_i^{(t)}(\theta_2)>0$ by Lemma \ref{lem:Lzeroornot}. By Proposition \ref{prop:update_t_g} for $\theta_1$ and $\theta_2$,
$$ \frac{g_i^{(t)}(\theta_1)}{g_i^{(t)}(\theta_2)} = \prod_{j=1}^{N} \Bigg( \frac{g_j^{(0)}(\theta_1)}{g_j^{(0)}(\theta_2)} \Bigg)^{P_{ji}^{(t)}}.$$
By Lemma \ref{lem:path_recurrence}, $P_{ji}^{(t)}$ is the $(j,i)$ entry of $A^t$.
By Theorem \ref{thm:perron}, this quantity equals
$$r^t \left( \frac{(vw^\top)_{ji}}{w^\top v} + \delta_{i,j}^{(t)} \right) = r^t \left( \frac{v_jw_i}{w^\top v} + \delta_{i,j}^{(t)} \right),$$ where $\delta_{i,j}^{(t)} \to 0$ as $t\to \infty$. The result follows by letting $\eps_{i,j}^{(t)} = \delta_{i,j}^{(t)} w^\top v/w_i$.

\subsection*{Proof of Proposition \ref{prop:maxdominate}}

Apply Proposition \ref{prop:ratio_of_f_asymptotic}.
As $t\to\infty$, the expression in the bracket converges to $L(\theta_1)/L(\theta_2)<1$, so this expression is smaller than some $\alpha<1$ for large $t$. By Lemma \ref{lem:spectral_radius_greater_than_one}, the exponent of this expression goes to infinity, so $f_i^{(t)}(\theta_1)/f_i^{(t)}(\theta_2)\to 0$ as $t\to\infty$.

\subsection*{Proof of Lemma \ref{lem:dominatedbound}}

All statements in this paragraph are meant to hold for almost every $x \in S$. For any $0\leq y<\delta_i/2$,
$$A(x) a_i(x)^y=A(x)^{1-y/\delta_i}(A(x) a_i(x)^{\delta_i})^{y/\delta_i}\leq A(x)^{1/2},$$
because $1-y/\delta_i>1/2$.
So if $\delta=\min \delta_i/2$, then
for any $0\leq y<\delta$, $A(x) a_i(x)^y\leq A(x)^{1/2}$ for all $i$.
For any $0 \leq y_i<\delta/N$,
$$\prod_{i=1}^N a_i(x)^{1+y_i}
=\prod_{i=1}^N (A(x) a_i(x)^{Ny_i})^{1/N}\leq A(x)^{1/2}.$$
Let $0<\delta'<1$ be such that
$(1+\delta')/(1-\delta')=1+\delta/N$. Then for any $\abs{y_i}< \delta'$, if $y_j = \min y_i$,
$$\prod_{i=1}^N a_i(x)^{1+y_i}
=\bigg( \prod_{i=1}^N a_i(x)^{(1+y_i)/(1+y_j)}\bigg)^{1+y_j}\leq A(x)^{(1+y_j)/2}\leq A(x)^{(1-\delta')/2},$$
because $1 \leq (1+y_i)/(1+y_j) < 1+\delta/N$.
Finally, because
$\prod_{i=1}^N a_i(x)^{1+y_i} \prod_{i=1}^N a_i(x)^{1-y_i}=A(x)^2$,
for any $\abs{y_i}< \delta'$,
$$A(x)^{(3+\delta')/2}\leq \prod_{i=1}^N a_i(x)^{1+y_i} \leq A(x)^{(1-\delta')/2}.$$

Let $T$ be such that for every $i$ and $t\geq T$, $\big\lvert\eps_i^{(t)}\big\rvert<\delta'$. Then, for $t\geq T$,
$$\int_{x\in S} c(x) A(x)^{r_t(3+\delta')/2} \leq I_t \leq \int_{x\in S} c(x) A(x)^{r_t(1-\delta')/2}.$$
Note that the right-hand expression may be infinite. 
Let
$$f_t(x)=c(x) A(x)^{r_t(3+\delta')/2},\quad
g_t(x)=c(x) A(x)^{r_t(1-\delta')/2}.$$
Then $f_T \in L^1$ and  $f_t, g_t\to c 1_{A(x)=1}$ pointwise a.e. as $t\to\infty$. For large $t$, $r_t(1-\delta')>r_T(3+\delta')$, so $f_t \leq g_t \leq f_T$ pointwise a.e.
By Lebesgue's dominated convergence theorem, $c1_{A(x)=1} \in L^1$ and $\int_{x} f_t, \int_{x} g_t\to \int_{A(x)=1} c(x)$ as $t\to\infty$,
so $I_t$ converges to this same limit.

\subsection*{Proof of Proposition \ref{prop:lessthanapoint}}

Apply Proposition \ref{prop:ratio_of_f_asymptotic} to get the expression into the form
$$\int_{\theta \in A} \frac{f_*(\theta)}{f_*(\theta')}
\bracket{
\prod_{j=1}^{N}
\Bigg( \frac{g_j^{(0)}(\theta)}{g_j^{(0)}(\theta')}\Bigg)^{v_j+\eps_{i,j}^{(t)}}}
^{w_i r^t/w^\top v}.$$
Then apply Lemma \ref{lem:dominatedbound} with $S=A$, $$c(\theta) = \frac{f_*(\theta)}{f_*(\theta')},\quad a_j(\theta) = \Bigg( \frac{g_j^{(0)}(\theta)}{g_j^{(0)}(\theta')}\Bigg)^{v_j}, \quad \eps_j^{(t)}=\frac{\eps_{i,j}^{(t)}}{v_j}, \quad r_t = \frac{w_i r^t}{w^\top v}.$$
The hypotheses of the lemma are readily verified using the hypotheses of the proposition and Lemma \ref{lem:spectral_radius_greater_than_one}.

\subsection*{Proof of Lemma \ref{lem:reversebound}}

By an analogous argument to the one in Lemma \ref{lem:dominatedbound} with reversed inequalities, there are $0<\delta'<1$ and $T$ such that, for $t\geq T$,
$$\int_{x\in S} c(x) A(x)^{r_t(3+\delta')/2} \geq I_t \geq \int_{x\in S} c(x) A(x)^{r_t(1-\delta')/2}.$$
The result is immediate if $A(x)=1$ almost everywhere. Otherwise, the right-hand integrand converges pointwise to a function that is infinite on a set of positive measure. By Fatou's lemma, the right-hand integral converges to infinity. Then $I_t \to \infty$ as $t \to\infty$.

\subsection*{Proof of Theorem \ref{thm:master}}

By Propositions \ref{prop:lessthanapoint} and \ref{prop:morethanapoint},
$$\frac{\int_{\theta\in A}f_i^{(t)}(\theta)}{\int_{\theta\in B}f_i^{(t)}(\theta)} \to \frac{M_1}{M_2},$$
where $M_1$ and $M_2$ are limits in the two respective propositions. Because $B$ has positive measure, $M_2>0$. Our final hypothesis implies that either $M_1=0$ or $M_2=\infty$, so $M_1/M_2=0$. Then our result follows from $\int_{\theta\in B}f_i^{(t)}(\theta) \leq 1$.

\subsection*{Proof of Proposition \ref{prop:nearessup}}

Note that $S$ has positive measure. Apply the Master Theorem \ref{thm:master}. Let $A=\Theta \setminus S$.
In the second case, pick any $\theta' \in T$ and let $B=T$.
Now consider the first case. If there is $\theta_1 \in S$ that minimizes $L$ in $S$,
then because $L$ is not constant a.e. on $S$, $\{\theta \in \Theta: L(\theta)> L(\theta_1)\}$ has positive measure.
By continuity of measure from below, there is $M'>L(\theta_1)$ such that $U=\{\theta \in \Theta: L(\theta)> M'\}$ has positive measure. Let $\theta'=\theta_1$ and $B=U$.
Otherwise, pick a sequence $\theta_i \in S$ such that $L(\theta_i) \to \inf_{\theta\in S} L(\theta)$. Define $S_i=\{\theta \in \Theta: L(\theta)\geq L(\theta_i)\}$.
Then $S=\cup_i S_i$.
By subadditivity of measure, one of $S_i$ has positive measure. There is $j$ such that $L(\theta_j)<L(\theta_i)$. Let $\theta'=\theta_j$  and $B=S_i$.

The result follows by applying the Master Theorem to $A$, $B$ and $\theta'$.
The existence of $\delta_i>0$ follows from the fact that $L(\theta)/L(\theta')\leq \alpha<1$ for all $\theta \in A$ and $g_i^{(0)}$ is bounded a.e. on $A$. The existence of $\delta'_i>0$ is similar.

\subsection*{Details of Example \ref{ex:abgraph}}
The weighted likelihood function is
$L(\theta)=g_x^{(0)}(\theta)^{1/2}g_y^{(0)}(\theta)^{1/2}$.
By Proposition \ref{prop:update_t_g},
$$f_x^{(t)}
=\frac{f_{\theta^*} \big(g_x^{(0)}g_y^{(0)}\big)^{(a+b)^t/2} \big( g_x^{(0)}/g_y^{(0)}\big)^{(a-b)^t/2}
}{\int_{\Theta} f_{\theta^*} \big(g_x^{(0)}g_y^{(0)}\big)^{(a+b)^t/2} \big(g_x^{(0)}/g_y^{(0)}\big)^{(a-b)^t/2}
}.$$

\subsection*{Details of Example \ref{ex:finite}}

$$f_x^{(t)}(0)=\frac{1}{1+((1-\alpha)/\alpha)^{(a-b)^t}}.$$

\subsection*{Proof of Theorem \ref{thm:Rkmain}}

If $\esssup_{|\theta|>r} L(\theta) <M< L_{\esssup}$, then $S_M$ has positive measure and $S_M \setminus B_r(0)$ has measure zero.
For any $\theta \in \theta_{\esssup}$ and $\delta >0$, $B_\delta(\theta) \cap S_M$ has positive measure, so $B_\delta(\theta) \cap B_r(0)$ has positive measure, so $\theta \in \overline{B_r(0)}$.
We conclude that $\theta_{\esssup} \subseteq \overline{B_r(0)}$ is bounded.

By Lemma \ref{lem:piecewisegood}, we can take $(U_i)_{i\geq 1}$ that works for all $g_i^{(0)}$. Assume that $U_1,\dots,U_n$ are the only ones whose closures intersect $\overline{B_r(0)}$. We claim that $\theta_{\esssup}=\cup_{i=1}^n \{\theta \in \overline{U_i}: \widetilde{L}_i(\theta)=L_{\esssup}\}$, where $\widetilde{L}_i$ is the continuous extension of $L$ on $\overline{U_i}$.
If $\theta \in \overline{U_i}$ is such that $\widetilde{L}_i(\theta)=L_{\esssup}$, then for every $M<L_{\esssup}$, there is $\delta'>0$ such that $B_{\delta'}(\theta) \cap U_i \subseteq S_M$ by continuity. So for any $\delta>0$, $B_{\delta}(\theta) \cap S_M \supseteq B_{\delta}(\theta) \cap B_{\delta'}(\theta) \cap U_i$ has positive measure because the latter is a nonempty open set. Hence $\theta \in \theta_{\esssup}$. Conversely, if $\theta \in \theta_{\esssup}$, then for every $M<L_{\esssup}$ and $\delta >0$,  $B_{\delta}(\theta) \cap S_M$ intersects some $U_i$. For $\delta<1$, $B_{\delta}(\theta) \subseteq B_1(\theta)$, so there are finitely many possible choices of $U_i$'s. If we pick sequences $M_j\to L_{\esssup}$ and $\delta_j \to 0$, some choice must repeat infinitely many times. Thus
$\theta \in \overline{U_i}$ for some $i$. 
Because $\theta \in \overline{B_r(0)}$, we have $1\leq i\leq n$.
Since we let $M\to L_{\esssup}$, $\widetilde{L}_i(\theta) \geq L_{\esssup}$. If $\widetilde{L}_i(\theta) > L_{\esssup}$, then there is a nonempty open set on which $L>L_{\esssup}$, a contradiction.
So $\widetilde{L}_i(\theta) = L_{\esssup}$, proving the claim. It follows from the claim that $\theta_{\esssup}$ is closed, so because it is bounded, it is compact.

Proposition \ref{prop:nearessup} implies that for any $M'<M<L_{\esssup}$, as $t\to\infty$,
$\int_{\theta \in S_M} f_i^{(t)}(\theta) \to 1$.
To prove the theorem, it suffices to show that there is $M'<M<L_{\esssup}$ such that $S_M\setminus U$ has measure zero. For large $M$, this is equivalent to $(S_M\cap B_r(0))\setminus U$ having measure zero, which is equivalent to $(S_M \cap U_i \cap B_r(0))\setminus U$ having measure zero for all $1\leq i\leq n$.
Suppose to the contrary that there is an increasing sequence $M_j \to L_{\esssup}$ such that $(S_{M_j} \cap U_i \cap B_r(0))\setminus U$ has positive measure. Pick $\theta_j \in (S_{M_j} \cap U_i \cap B_r(0))\setminus U$. Then $\theta_j$ is a bounded sequence in $U_i$ such that $L(\theta_j) \to L_{\esssup}$. We can pass to a subsequence and assume that $\theta_j \to \theta \in \overline{U_i}$. Then $\theta \in \theta_{\esssup}$, so $\theta_j \in U$ for large $j$, a contradiction.
Therefore the desired integral converges to 1.
Note that the nonemptiness of $\theta_{\esssup}$ is part of the conclusion.

To see that the distribution converges to the point distribution at $\theta$ if $\theta_{\esssup}=\{\theta\}$, take $U=B_{\delta}(\theta)$ for small $\delta >0$.

\subsection*{Proof of Proposition \ref{prop:pois}}

The initial condition is feasible by Corollary \ref{cor:suffcientfeas}. We have
$$L(\theta) \propto \frac{(\lambda^*)^{\theta}}{\theta!}.$$
By direct computation, $L$ has a unique maximum at $\lfloor \lambda^* \rfloor$ if $\lambda^*$ is not an integer, and two maxima $\{\lambda^*-1,\lambda^*\}$ if $\lambda^*$ is an integer. The proposition follows from Theorem \ref{thm:maincountable}.

\subsection*{Proof of Proposition \ref{prop:gaussian}}

The weighted likelihood function
$$L(\theta) \propto \exp \left[ - \frac{1}{2}\sum_{i=1}^{N} v_i \tau_i (\theta- \mu_i)^2 \right]$$
is proportional to a Gaussian with the maximum at $\theta_{\max}$.
The result follows from Corollary \ref{cor:Rkcontinuous}.

\end{shownto}

\end{document}